\documentclass[11pt]{amsart}
\usepackage{amsfonts,amssymb,graphicx,mystyle,color}
\usepackage[colorlinks,linkcolor=blue]{hyperref} 
\usepackage{cite}
\usepackage{hyperref}
\usepackage{tikz}
\usepackage{bbm}
    
\def\a{\alpha} \def\b{\beta}  \def\e{\varepsilon}  \def\g{\gamma} \def\l{\lambda} \def\s{\sigma} \def\t{\tau} 
\def\D{\Delta}   
\def\S{\Sigma}
    \def\H{\mathcal H} 

 \def\T{\mathcal T}  
\def\Re{\mathbb{R}}   
\def\O{\Omega}

\def\k{\kappa}
\def\rho{\varrho}

\def\supp{\text{supp}}
\def\Cav{\text{Cav}}

\begin{document}\openup 1\jot
\title[Information Spillover]{Information spillover in multiple Zero-sum Games\\}
 \author[L. Pahl]{Lucas Pahl}\thanks{I am grateful to Paulo Barelli and Hari Govindan for their guidance and encouragement. I would like to thank Sven Rady, Rida Laraki, Tristan Tomala, Heng Liu, Mathijs Janssen and two anonymous referees for comments and suggestions. I acknowledge financial support from the Hausdorff Center for Mathematics (DFG project no. 390685813).}
 \address{Hausdorff Center for Mathematics and Institute for Microeconomics, University of Bonn, Adenauerallee 24-42, 53113 Bonn, Germany.}
 \email{pahl.lucas@gmail.com}
 \date{\today. This paper subsumes a previous paper titled ``Information Spillover in Bayesian Repeated Games''.}
\maketitle

\begin{abstract} This paper considers an infinitely repeated three-player zero-sum game with two-sided incomplete information, in which an informed player plays two zero-sum games simultaneously at each stage against two uninformed players. This is a generalization of the model in Aumann et al. \cite{AMS1995} of two-player zero-sum one-sided incomplete information games. Under a correlated prior, the informed player faces the problem of how to optimally disclose information among two uninformed players in order to maximize his long-term average payoffs (i.e., undiscounted payoffs). Our objective is to understand the adverse effects of ``information spillover'' from one game to the other in the equilibrium payoff set of the informed player. We provide conditions under which the informed player can fully overcome such adverse effects and characterize equilibrium payoffs.  In a second result, we show how the effects of information spillover on the equilibrium payoff set of the informed player might be severe. Finally, we compare our findings on the equilibrium-payoff set of the informed player with those of Bayesian Persuasion models with multiple receivers. 
\end{abstract}

\maketitle

\section{Introduction}

In their seminal work, Aumann et al. \cite{AMS1995} analyzed an undiscounted infinitely repeated game with one-sided incomplete information: one player (the informed) knows the stage game being played whereas the other (the uninformed) does not know and cannot observe payoffs, only actions. They showed that this game has a value and constructed optimal strategies for the players. Matters are more complicated if the informed player were to play against {\it more than one} uninformed player, as it would be the case of a military power (e.g., USA) negotiating with two different countries (e.g., Russia and Iran).\footnote{The USA may want to conceal from Russia the exact size of its arsenal, and at the same time may want to reveal it to Iran to leverage its bargaining position. More examples in this line can be found in Aumann et al. \cite{AMS1995}.} By observing what the informed player plays against some other uninformed player, an uninformed player can make inferences about the game he plays against the informed player. As a consequence, it may not be optimal for the informed player to play his unilaterally optimal strategy against some of the uninformed players. Put differently, the {\it information spillover} among the games played between the informed player and the uninformed players adds layers of complexity to the analysis. 

We consider a three-player undiscounted infinitely repeated game in which one of the players is informed of the two zero-sum stage games that he plays against each of the other two (uninformed) players. Each uninformed player only knows the prior probability distribution over the finite set of pairs of zero-sum finite-action stage games, and during the play of the game observes the profiles of actions (but not the payoffs). The informed player collects the sum of payoffs from the two component games. 

In the absence of information spillover, for instance when each uninformed player cannnot observe the actions played in the other zero-sum game, our three-player game has a single expected payoff, namely, the sum of values of each of the two-player component games. However, when all players are able to observe the actions played across each zero-sum game, the information spillover kicks in and it is in principle unclear whether the informed player can attain the sum of values in equilibrium. This sum of values can actually be shown to be an upper bound on the equilibrium payoffs of the informed player in our three-player game.

Our first main result has two parts which together provide a condition under which the informed player can attain this upper bound in equilibrium, even in the presence of information spillover. More precisely, under such condition, we show the informed player can attain anything as an equilibrium payoff from his individually rational payoff to the above mentioned upper bound, thus characterizing the set of equilibrium payoffs in the three-player game.  In particular, this result implies that the three-player model we analyse might have a continuum of equilibrium payoffs, even though it is a zero-sum model. 

The method used to obtain this first result is also of interest to the model of two-player games studied by Aumann et al. \cite{AMS1995}. Under a sufficient condition on the stage payoffs, we show that different optimal strategies from those constructed by Aumann et al. exist.\footnote{We refer here to the so-called ``splitting strategies'' where the informed player first signals information about the underlying stage game, so as to ``concavify'' the nonrevealing value function.} The strategy of the informed player, in particular, does not involve any signalling on path of play, even when the standard optimal strategy constructed by Aumann et al. necessarily does. 

In a second result we provide a necessary condition for the existence of equilibria yielding the upper bound to the informed player. We explore two consequences of this result. First, we show that a natural class of equilibria which involve signalling on equilibrium path never pays the upper bound to the informed player. Second, we present an example showing that the effects of information spillover might be very severe, in the sense that the informed player is not able to attain the upper bound in equilibrium.

Given the recent heightened interest in Bayesian Persuasion (BP) (since \cite{KG2011}’s seminal contribution) and its overtones with Aumann et al., it is interesting to compare the effects of information spillover in a multiple-receiver setting in BP with the results of this paper. We consider a model with one sender and two receivers under two alternative specifications: (i) the sender can send private messages to each receiver and (ii) the sender can only send public messages. We show that the difference in the sender’s payoff from (i) to (ii) can be interpreted as a loss due to information spillover, but a similar intuition cannot be extended to our repeated-game model.


\subsection{Related Literature}To the best of our knowledge, the model analyzed here is new. Although the model we analyse is zero-sum, the results and the techniques presented remain closer to the non-zero-sum literature, especially to Hart \cite{SH1985} and Sorin \cite{SS1983}. We highlight here a few additional papers on discounted and undiscounted repeated games with incomplete information that have technical and thematic similarities to this project. A significant part of the literature on undiscounted repeated games with incomplete information analyses models under the assumption of ``known own payoffs'' (see Forges \cite{FF1992}).  This is a reasonable assumption in several applications and allows for equilibrium-payoff characterizations which are especially tractable (see  Shalev \cite{JS1988}). Under this assumption, Forges and Solomon \cite{FS2015} provide a simple characterization of Nash equilibrium payoffs in undiscounted repeated games of two players and incomplete information.\footnote{An additional assumption needed for the characterization is the existence of ``uniform punishment strategies'' for the players in the stage game, that is, strategies that allow a player to be punished by holding his payoffs at his ex-post individually rational level.} This characterization is used to show that in a class of public good games, uniform equilibria might not exist. More closely related to our paper in terms of the information environment is Forges et al. \cite{FHS2016}. In this paper, among other results, cooperative solutions of one-shot  games with two players and exactly one informed player are related to noncooperative solutions of two-player repeated games with exactly one informed player. More specifically, under the assumption of existence of uniform punishment strategies for the uninformed player, the \textit{joint plan} equilibrium payoffs of the repeated game equal the set of cooperative solutions of the one-shot game. This \textit{folk theorem} is not however an equilibrium payoff characterization, since it is known from Hart \cite{SH1985} that joint plans cannot account for the whole of equilibrium payoffs in general.

Our work was inspired by Huangfu and Liu \cite{HL2021} who considered the issue of information spillovers between markets. In their model, a seller holds private information about the quality of goods he sells in two different markets and buyers learn about the seller's private information from observing past trading outcomes not only in the market in which they directly participate, but also from observing the outcomes of the other market. The authors show that, under certain assumptions on the correlation of qualities between goods in different markets, information spillover mitigates the negative effects of adverse selection. 

The literature on Bayesian Persuasion with multiple receivers is extensive. We highlight here a few papers that connect to our discussion in the last section. Wang \cite{YW2013} presents a Bayesian Persuasion model with multiple Receivers whose payoffs depend not only on the unobserved state drawn by Nature but also on the action choices of the other receivers. Concretely, the Receivers vote on the outcome after receiving a signal about the underlying state. The paper compares two environments: one in which all Receivers observe a public message drawn from the experiment chosen by the Sender, and another in which each Receiver observes a private message independently drawn from the same experiment set out by the Sender. The main difference with the models we analyse in our last section is the payoff-interdependencies of the receivers: in our models Receivers do not care about each other's action choices, but only about their own. This implies that there is in effect no game among Receivers in our model, only a decision problem. The payoff interdependencies among Receivers coupled with the assumption of the independence of the messages drawn from the Sender's experiment drive the result of that paper in which the best equilibrium payoffs for the Sender are higher in the public compared to the private message environment, a result that is the opposite of what we find. 
	
Arieli and Babichenko \cite{AB2019} are interested in what happens with optimal experiments under different assumptions for the utility of the Sender as well as payoff interdependencies of the Receivers. One result is particularly reminiscent of the one we obtain in our analysis: a specification of the public Bayesian Persuasion model we define in our paper (say, with two perfectly correlated states for each Receiver and a well-chosen utility function for the Sender) yields \cite{AB2019}'s model, where the utility of the Sender is additive over the Receivers, and therefore supermodular in particular. Theorem 3 in that paper then applies and a public signaling experiment is optimal if and only if all Receivers have the same ``persuasion level'', i.e., are essentially identical from the point of view of the optimal persuasion policy. Therefore, if Receivers meaningfully differ, the optimal experiment must send private messages. In our words, this result could be read as saying that unless information spillover is not meaningful (i.e., Receivers have identical persuasion levels), public Bayesian Persuasion yields lower a equilibrium payoff for the Sender compared to private Bayesian Persuasion. Finally, Koessler, Laclau, and Tomala \cite{KLT2022} generalize the standard Bayesian Persuasion model in a number of directions, most importantly by considering multiple Senders and multiple Receivers.


\subsection{Organization}The remainder of the paper is organized as follows. Section \ref{SEC: Model} presents the model. Our first main result is divided in 2 parts: main result 1 and main result 2. Section \ref{SEC: EPC} presents our main result 1. Section \ref{SEC: empty} presents our main result 2. Section \ref{BP} compares the equilibrium payoff set in our model with that of a Bayesian Persuasion model with multiple receivers and Section \ref{SEC: conc} concludes. The proofs of technical results are left to the Appendix. Additional results can be found in the Supplemental Appendix.

\section{Model and Equilibrium Concept}\label{SEC: Model}

\subsection{Notation}Given a finite set $K$, $\Delta(K)$ is the set of probability distributions over $K$; given a topological space $X$  the interior of $X$ will be denoted by $\text{int}(X)$. If $X$ is a subspace, its boundary will be denoted by $\partial(X)$. For a set $Y \subseteq \Re^{m}$, its convex closure is denoted by $\text{co}(Y)$.  For $p \in \Delta(K_A \times K_B)$, $p_A$ (resp. $p_B$) denotes its marginal on $K_A$ (resp. $K_B$), and $\supp(p)$ its support. We denote a product distribution on $K_A \times K_B$ by $p_A\bigotimes p_B$, and use $\Delta(K_A) \bigotimes \Delta(K_B)$ to denote the set of all such distributions.

\subsection{Model}A \textit{three-player infinitely repeated zero-sum game with two-sided incomplete information}, denoted $\mathcal{G}(p^0)$, is given by the following data: 
\begin{itemize}

\item Three players, namely player 1 (the informed player), player 2 and player 3 (the uninformed players).
\item Finite sets: $I_i$, $J_i$, $K_i$, $i=A,B$  with $I_A\times I_B$ (resp. $J_A$ and $J_B$) being the set of actions of player 1 (resp. players 2 and 3), and $K_A\times K_B $ being the set of states.
\item $p^0 \in \Delta(K_A \times K_B)$ is the prior.
\item For each $k_A \in K_A$ and $k_B \in K_B$,  $A^{k_A}$ and $B^{k_B}$ matrices of dimensions $|I_A| \times |J_A|$ and $|I_B| \times |J_B|$ respectively. $A^{k_A}$ and $B^{k_B}$ are the stage-game payoff matrices. 
\end{itemize}

The play of the infinitely repeated game is as follows: 

\begin{itemize}
\item At stage 0, state $(k_A, k_B) \in K_A \times K_B$ is drawn according to distribution $p^0$ and only player 1 knows the draw.  
\item At each stage $t=1,2,...$, the players independently choose an action in their own set of actions: player $1$ chooses $(i^t_A, i^t_B) \in I_A \times I_B$ and players $2$ and $3$ choose $j^t_A \in J_A$ and $j^t_B \in J_B$, respectively. The stage payoff to player 1 is then  $A^{k_A}_{i^t_A,j^t_A} + B^{k_B}_{i^t_B, j^t_B}$; to player 2, $-A^{k_A}_{i^t_A, j^t_A}$ and to player 3, $-B^{k_B}_{i^t_B, j^t_B}$. Monitoring is perfect, i.e., the chosen actions are observed by all players before starting stage $t+1$. Realized payoffs are not observed by the players (though player 1 knows them, since he is fully informed).
\end{itemize}

Players are assumed to have perfect recall and the whole description of the game is common knowledge. 

A behavior strategy for player 1 is a tuple $\s = (\s^{(k_A,k_B)})_{(k_A,k_B) \in K_A \times K_B}$, where for each $(k_A,k_B) \in K_A \times K_B$, $\s^{(k_A,k_B)} = (\s^{(k_A,k_B)}_t)_{t \geq 1}$ and $\s^{(k_A,k_B)}_t$ is a mapping from the Cartesian product $H_t := (I_A \times J_A \times I_B \times J_B)^{t-1}$ (with $H_0 := \{\emptyset\}$) to $\D(I_A \times I_B)$, giving the lottery on actions played by player 1 at a stage $t$, when the state is $(k_A, k_B)$. Because players 2 and 3 do not know the state, a behavior strategy for player 2 (resp. player 3) is an element $\t_A = (\t_{A,t})_{t \geq 1}$ (resp. $\t_B = (\t_{B,t})_{t \geq 1}$), where $\t_{A,t}$ (resp. $\t_{B,t}$) is a mapping from $H_t$ to $\D(J_A)$ (resp. $\D(J_B))$, giving the lottery on actions to be played by player 2 (resp. player 3) on stage $t$. The set of behavior strategies of player 1 is denoted by $\S$; for player 2, it is denoted by $\T_A$ and for player 3,\ it is denoted by $\T_B$. 


A behavior strategy profile $(\sigma, \tau_A, \tau_B)$ induces, for every state $(k_A, k_B)$ and stage $T>0$, a probability distribution on $H_{T+1}$. Also, $(\sigma, \tau_A, \tau_B)$ and $p^0$ induce a probability distribution over $K_A \times K_B \times H_{T+1}$. We can thus define the expected average payoffs (with $\kappa$ being a random variable taking values on $K_A \times K_B$ distributed according to $p^0$, $\kappa_A$ the random variable obtained from projecting $\kappa$ on $K_A$ and $\kappa_B$ the random variable obtained from projecting $\kappa$ on $K_B$):


 $$\a^{k_A, k_B}_T = \alpha^{k_A, k_B}_T (\sigma,\tau_A,\tau_B) := \mathbb{E}_{\sigma^{(k_A, k_B)},\tau_A, \tau_B}\Big[\frac{1}{T} \sum^{T}_{t=1}(A^{k_A}_{i^t_A,j^t_A} + B^{k_B}_{i^t_B, j^t_B})\Big], $$

$$\beta^{A}_T(\sigma, \tau_A, \tau_B) :=  \mathbb{E}_{\sigma,\tau_A, \tau_B,p^0}\Big{[}\frac{1}{T} \sum^{T}_{t=1}(-A^{\kappa_A}_{i^t_A, j^t_A})\Big{]},$$

 $$\beta^{B}_T(\sigma, \tau_A, \tau_B) : = \mathbb{E}_{\sigma,\tau_A, \tau_B,p^0}\Big{[}\frac{1}{T} \sum^{T}_{t=1}(-B^{\kappa_B}_{i^t_B, j^t_B})\Big{]}.$$

The number $\alpha^{k_A, k_B}_T (\sigma,\tau_A,\tau_B)$ is the expected average payoff (up to time $T$) of player 1; $\beta^{A}_T(\sigma, \tau_A, \tau_B)$ is the expected average payoff (up to time $T$) of player 2 and $\beta^{B}_T(\sigma, \tau_A, \tau_B)$ is the expected average payoff (up to time $T$) of player 3.

The model defined is therefore a ``combination'' of two zero-sum games: at each stage, player 1 plays simultaneously a zero-sum game against player 2 and another zero-sum game against player 3, collecting the sum of the payoffs of each of these games; players 2 and 3 are the minimizers in each of the zero-sum game they play against player 1. One distinctive and important aspect of this model is the fact that each uninformed player can observe not only the actions played in his own zero-sum game, but also the actions played in the other zero-sum game.

\subsection*{Equilibrium Concept}

A profile $(\sigma, \tau_{A}, \tau_{B})$ is a \textit{uniform equilibrium} of $\mathcal{G}(p^0)$ when:

\begin{enumerate}

\item For each $(k_A, k_B) \in \text{supp}(p^0)$,\footnote{Remark \ref{important remark} in Section \ref{SEC: EPC} shows that assuming the prior $p^0 \in \text{int}(\D(K_A \times K_B))$ (as customary in the literature) is not without loss of generality for the results in this paper. This is why we present the definition requiring convergence of $(\a^{k_A, k_B}_T)_{T \geq 1}$, $(k_A, k_B) \in \text{supp}(p^0)$. The same reasoning applies to condition (2).} $(\alpha^{k_A, k_B}_T (\sigma,\tau_A,\tau_B))_{T \geq 1}$ converges as $T$ goes to infinity to some $\alpha^{k_A, k_B}(\sigma,\tau_A,\tau_B)$, $(\beta^{A}_T(\sigma, \tau_A, \tau_B))_{T \geq 1}$ converges to some $\beta^A(\sigma, \tau_A, \tau_B)$ and $(\beta^{B}_T(\sigma, \tau_A, \tau_B))_{T \geq 1}$ converges to some $\beta^B(\sigma, \tau_A, \tau_B)$.

\item For each $\epsilon >0$, there exists a positive integer $T_0$ such that for all $T \geq T_0$, $(\sigma,\tau_A, \tau_B)$ is an $\epsilon$-Nash equilibrium in the finitely repeated game with $T$ stages, i.e., \begin{enumerate} 
\item For each $(k_A, k_B) \in \text{supp}(p^0)$ and $\s' \in \S$, $\alpha^{k_A, k_B}_T(\sigma', \tau_A, \tau_B) \leq \alpha^{k_A,k_B}_T(\sigma, \tau_A, \tau_B) + \epsilon$; 
\item For each $\tau_A' \in \mathcal{T}_A$, $\beta^{A}_T(\sigma,\tau_A', \tau_B) \leq \beta^{A}_T(\sigma,\tau_A, \tau_B) + \epsilon$;
\item For each $\tau_B' \in \mathcal{T}_B$,  $\beta^{B}_T(\sigma,\tau_A, \tau_B') \leq \beta^{B}_T(\sigma,\tau_A, \tau_B) + \epsilon$.
\end{enumerate}

\end{enumerate}

Uniform equilibrium is a standard equilibrium concept for the analysis of undiscounted repeated games. It contains a strong requirement, namely (2), which posits that the profile $(\s,\t_A,\tau_B)$ must generate an $\epsilon$-Nash equilibrium on the finite  but sufficiently long-horizon ($T \geq T_0$, where $T < \infty$) version of our model.\footnote{One notable aspect of this equilibrium notion is that uniform equilibria in our model are approximate Nash equilibria in the discounted version of our model: if $(\s, \t_A, \t_B)$ is a uniform equilibrium, then $(\s, \t_A, \t_B)$ is a $\e$-Nash equilibrium of the discounted versions of our model for a sufficiently high discount factor. See Theorem 13.32 in Maschler et al. \cite{MSZ2013}.}

Unless explicitly stated otherwise, from now on whenever we refer to uniform equilibrium or equilibria we will use simply \textit{equilibrium} or \textit{equilibria}.

If $(\sigma, \tau_A, \tau_B)$ is an equilibrium in $\mathcal{G}(p^0)$, the associated vector $$(\alpha(\sigma, \tau_A, \tau_B), \beta^A(\sigma, \tau_A, \tau_B), \beta^B(\sigma, \tau_A, \tau_B)),$$ where $\alpha(\sigma,\tau_A, \tau_B) := (\alpha^{k_A, k_B}(\sigma, \tau_A, \tau_B))_{(k_A,k_B) \in \text{supp}(p^0)}$, is the \textit{vector of payoffs of}  $(\sigma, \tau_A, \tau_B)$. Also  $\alpha(\sigma,\tau_A, \tau_B) \cdot p^0$ (where $\cdot$ is the standard scalar product in Euclidean space) is the \textit{ex-ante equilibrium payoff} of the informed player. 

\subsection{Preliminaries on the Aumann et al. model}\label{sec:AM}Our analysis of the equilibrium payoff set of the game $\mathcal{G}(p^0)$ in the next section will rely on certain properties of each of the two-player, infinitely repeated zero-sum games that the informed player plays against each uninformed player. For this reason we now recall some of the main results in Aumann et al. \cite{AMS1995}, which is the original reference for this two-player model. Let $K$ be the finite set of states and $M = (M^k)_{k \in K}$ a collection of zero-sum payoff matrices where $M^{k} \in \mathbb{R}^{I \times J}$ for each $k \in K$. Denote by $G_M(p)$ the \textit{infinitely repeated, two-player, zero-sum game with one-sided incomplete information} with prior $p \in \Delta(K)$ and undiscounted payoffs (see Sorin \cite{SS2002}, Chapter 3, for a detailed description of this model or Aumann et al. \cite{AMS1995}). Let $M(p) = \sum_{k \in K} p^k M^k$ and define $v_{M}(p) = \text{min}_{t \in \D(J)}\text{max}_{s \in \D(I)}s M(p) t = \text{max}_{s \in \D(I)}\text{min}_{t \in \D(J)}s M(p) t$, where $s$ is a row vector and $t$ a column vector. The function $q \in \D(K) \mapsto v_{M}(q) \in \Re$ is called the {\it non-revealing value function}. Let $\text{Cav}(v_M)$ be the (pointwise) smallest concave function $g$ from $\Delta(K)$ to $\mathbb{R}$ such that $g(q) \geq v_M(q)$ for all $q \in \Delta(K)$. Alternatively, one can define $\text{Cav}(v_M)(q) := \text{sup} \{\sum^k_{i=1}\alpha_i v_M(q_i) | \exists k \in \mathbb{N}, \forall i \in \{1,...,k\}, \alpha_i \geq 0,  \sum^k_{i=1}\alpha_i q_i = q, \sum^k_{i=1}\alpha_i =1\}$. Aumann et al. \cite{AMS1995} proved that a (uniform) value of $G_M(p)$ exists and equals $\text{Cav}(v_M)(p)$. They also showed how to construct (uniformily) optimal strategies for both players. 

Given the model $\mathcal{G}(p^0)$, the infinitely repeated, two-player, zero-sum game with one-sided incomplete information with prior $p^0_A$ defined by states $K_A$ and payoff matrices $(A^{k_A})_{k_A \in K_A}$ with undiscounted payoffs will be denoted $G_A(p^0_A)$ -- this game is played by players 1 (informed) and 2 (uninformed). Analogously, we define $G_B(p^0_B)$ as the two-player, infinitely repeated, zero-sum game with one-sided incomplete information played between players 1 and 3. The two-player, infinitely repeated zero-sum game with one-sided incomplete information and prior $p^0 \in \Delta(K_A \times K_B)$ with undiscounted payoffs, where stage payoff matrices are $(C^{k_A, k_B})_{k_A \in K_A, k_B \in K_B}$ given by $C^{k_A, k_B}_{i_A,i_B,j_A,j_B} := A^{k_A}_{i_A,j_A} + B^{k_B}_{i_B, j_B}$ will be denoted $G_{A+B}(p^0)$. This two-player game will be used as an auxiliary game to construct strategies in the three-player game $\mathcal{G}(p^0)$.

\subsection{Example}\label{Example 1}
We would like to illustrate the new strategic difficulties  that arise in the model $\mathcal{G}(p^0)$ in comparison to the two-player zero-sum model of Aumann et al. \cite{AMS1995}. Specifically, we would like to show that in this example that if the informed player plays the optimal strategies constructed by Aumann et al. in each game $G_A(p^0_A)$ and $G_B(p^0_B)$, he cannot guarantee the  ex-ante expected payoff equal to the sum of values Cav$(v_A)(p^0_A)$ + Cav$(v_B)(p^0_B)$. 

Two sets $A = \{A^1,A^2\}$ and $B = \{B^1,B^2\}$ of payoff matrices are defined below together with $p^0 \in \Delta(K_A \times K_B)$, where $K_A = \{1,2\}$ and $K_B = \{1,2\}$.

$$p^0 = \begin{bmatrix} 1/2 & 0 \\ 0 & 1/2 \end{bmatrix}$$
$$\\
A^1 = \begin{bmatrix} 1 & 0 \\ 0 & 0  \end{bmatrix} \,\, A^2 = \begin{bmatrix} 0 & 0 \\ 0 & 1 \end{bmatrix}
$$
$$
B^1 = \begin{bmatrix} 4 & 0 & 2 \\ 4 & 0 & -2 \end{bmatrix} \,\,B^2 = \begin{bmatrix}0 & 4 & -2 \\ 0 & 4 & 2 
\end{bmatrix} 
$$\\

In the matrix $p^0$, an entry $p^0_{ij}$ corresponds to the probability with which Nature chooses $A^{i}$ and $B^{j}$. So, the stage-payoffs in the two zero-sum games are given by $A^1$ and $B^1$ with probability $1/2$, and they are given by $A^2$ and $B^2$ with probability $1/2$. Since the prior assigns perfect correlation between states, there are only two states to consider effectively: states $(1,1)$ and $(2,2)$. 

In figures \ref{FigCavvA} and \ref{FigCavvB}, $q$ denotes the probability of state $(1,1)$ and $1-q$ the probability of state $(2,2)$. Each row of $A^i$ and $B^j, i,j \in \{1,2\}$ corresponds to a stage-game action of the informed player: call the first row ``$U$'' and the second row ``$D$''. 
By computation, we get that the non-revealing values are:

\[v_A(q)=q(1-q)\text{ for all }q \in [0,1]\]  

\[v_B(q) = \begin{cases}
       4q        &\text{if }  q \in [0,1/4)\\
       -4q +2    &\text{if }  q \in [1/4,1/2)\\
       4q -2    &\text{if }  q \in [1/2,3/4)\\
       -4q + 4 &\text{if }  q \in [3/4,1]
    \end{cases}\]
 
These imply that the concavification of these values are: 

\[\text{Cav}(v_A)(q)= v_A(q)=q(1-q)\text{ for all }q \in [0,1]\]

\[\text{Cav}(v_B)(q) = \begin{cases} 
       4q & \text{if } q \in [0,1/4)\\
       1 & \text{if } q \in [1/4,3/4)\\
       -4q + 4 & \text{if } q \in [3/4,1]
     \end{cases}\]

\begin{figure}[t]
\centering{}%
\caption{Graphs of $\text{Cav}(v_A)$ and $v_A$}
\begin{tikzpicture}\label{FigCavvA}

  \draw[->] (0,0) -- (4.5,0) node[right] {$q$};
  \draw[->] (0,-1) -- (0,2) node[above] {$\mathbb{R}$};
  \draw[scale=4,domain=0:1,smooth,variable=\p,blue, thick] plot ({\p},{\p*(1-\p)});
 
\node at (4,-0.25) (noderight) {1};
\node at (-0.25,-0.25) (noderight) {0};

\end{tikzpicture}

\caption{Graphs of $\text{Cav}(v_B)$ and $v_B$}
\begin{tikzpicture}\label{FigCavvB}

 \draw[->] (0,0) -- (4.5,0) node[right] {$q$};
  \draw[->] (0,-1) -- (0,4) node[above] {$\mathbb{R}$};
  \draw[scale=4,domain=0:0.25,smooth,variable=\p,blue, thick] plot ({\p},{4*(\p)});
  \draw[scale=4,domain=0.25:0.5,smooth,variable=\y,blue, thick ] plot ({\y},{-4*(\y)+2});
  \draw[scale=4,domain=0.5:0.75,smooth,variable=\y,blue, thick] plot ({\y},{4*(\y)-2});
  \draw[scale=4,domain=0.75:1,smooth,variable=\y,blue, thick] plot ({\y},{-4*(\y)+4});
  \draw[scale=4,domain=0.25:0.75, densely dotted, variable=\s,red, thick] plot ({\s},{1});
  
\node at (3,-0.35) (noderight) {$\frac{3}{4}$};
\node at (1,-0.35) (noderight) {$\frac{1}{4}$};
\node at (2, -0.35) (noderight) {$\frac{1}{2}$};
\node at (4,-0.25) (noderight) {1};
\node at (-0.25,-0.25) (noderight) {0};

\end{tikzpicture}
\end{figure} 

\bigskip

We present the optimal strategy of the informed player in the game $G_B(1/2)$. The optimal strategy of the informed player in the two-player repeated zero-sum game $G_B(1/2)$ is defined as follows:  in case the state drawn by Nature is 1, the informed player plays ``$U$''  with probability $1/4$ and after that plays at each stage, independently, the optimal action of the one-shot zero-sum game whose matrix is $B(1/4)$; in case the state drawn by Nature is 2, the informed player plays ``$U$'' with probability $3/4$ and, after that, plays the optimal action of the one-shot zero-sum game whose matrix is $B(3/4)$. 

After observing the realized action of the informed player in the first stage, the uninformed player updates his beliefs about the states to \textit{posteriors} about states 1 and 2: in our example, the uninformed player, after observing ``$U$'', assigns probability $1/4$ to the state being 1. After observing ``$D$'', the uninformed player assigns probability $3/4$ to the state being $1$. The strategy just described for the informed player is an example of a \textit{signalling strategy}: the informed player uses his actions to signal information about the underlying state. 

After the first stage, according to the construction described, no more information is signaled and the uninformed player plays the optimal action of the one-shot zero-sum game with matrix $B(1/4)$ or $B(3/4)$ forever, depending on whether $U$ or $D$ was realized, respectively. Playing the signalling strategy guarantees to the informed player an ex-ante payoff of $(1/2) v_B(1/4) + (1/2) v_B (3/4) = (1/2) \text{Cav}(v_B)(1/4) + (1/2) \text{Cav}(v_B)(3/4) = \text{Cav}(v_B)(1/2) = 1$. 
	
Now, in game $G_A(1/2)$ the non-revealing value function of the informed player is strictly concave, which implies that his optimal strategy in this game is non-revealing (at any prior): one optimal strategy for the informed player is to play at each stage the optimal action of the one-shot zero-sum game with matrix $A(1/2)$ independently forever, which generates no uptading of the beliefs on the part of the uninformed players. 

If the informed player uses the signalling strategy described in $G_B(1/2)$, because of perfect correlation between $\kappa_A$ and $\kappa_B$, this strategy also induces the same updating on the part of the (uninformed) player 2, inducing,  similarly, posteriors $1/4$ and $3/4$ in the two-player zero-sum repeated game $G_A(1/2)$. By using that strategy in game $G_B(1/2)$, the informed player in game $G_A(1/2)$ can now only guarantee $(1/2)\text{Cav}(v_A)(1/4) + (1/2) \text{Cav}(v_A)(3/4) = 3/16 < 1/4 = \text{Cav}(v_A)(1/2)$. 

Therefore, if the informed player plays the strategies described above in games $G_A(1/2)$ and $G_B(1/2)$, he cannnot guarantee in $\mathcal{G}(p^0)$ the sum of the uniform values of each zero-sum game i.e., Cav$(v_A)(p^0_A)$ $+$ Cav$(v_B)(p^0_B)$. 

\section{Main Result 1: Equilibrium Payoff-Set Characterization}
\label{SEC: EPC}

Our first main result (main result 1) has two parts: The first part of main result 1 is Theorem \ref{mainthm1} which provides a sufficient condition under which  the ex-ante equilibrium payoffs of the informed player in $\mathcal{G}(p^0)$ permits a simple characterization. The second part of our main result 1 is Theorem \ref{main1}, which provides a general class of games under which that sufficient condition holds. The two results provide a condition under which $\mathcal{G}(p^0)$ has a continuum of equilibrium payoffs.

This section is subdivided in three subsections. The first of these, subsection \ref{sec:MR1a}, is dedicated to Theorem \ref{mainthm1} and the main ideas of its proof. Subsection \ref{sec:MR1b} is dedicated to Theorem \ref{main1}. Finally, subsection \ref{SEC: 2-player games} highlights a by-product for the theory of two-player zero-sum games of the equilibrium constructions used in Theorems \ref{mainthm1} and \ref{main1}.


\subsection{First Part of Main Result 1}\label{sec:MR1a}

We now introduce the necessary concepts and state our main result 1 in full generality. Paralleling the definitions of the previous section, we denote the set of histories at stage $t \geq 1$ for a two-player, zero-sum, infinitely repeated, undiscounted game with one-sided incomplete information $G_A(p^0_A)$ by $H^A_t$ with generic element $h^A_t$. The notation for a behavior strategy of the informed player in $G_A(p^0_A)$ is exactly analogous to the one defined for player 1 in the three-player game $\mathcal{G}(p^0)$: $\s_A = (\s^{k_A}_A)_{k_A \in K_A}$ and  $\s^{k_A}_A = (\s^{k_A}_{At})_{t \geq 1}$, with $\s^{k_A}_{At} : H^A_t \to \D(I_A)$.

\begin{definition} 
 An equilibrium $(\s_A, \t_A)$ of $G_A(p^0_A)$ is \textit{non-revealing} if for each $t \geq 1$, $k_A, k'_A \in \supp(p^0_A)$ and $h^A_t \in H^A_t$ played with positive probability by $(\s_A, \t_A)$, we have $\s^{k_A}_{At}(h^A_t) = \s^{k'_A}_{At}(h^A_t)$. Analogously, an equilibrium $(\s_B, \t_B)$ of $G_B(p^0_B)$ is \textit{non-revealing} if for each $t \geq 1$, $k_B, k'_B \in \supp(p^0_B)$ and $h^B_t \in H^B_t$ played with positive probability by $(\s_B, \t_B)$, we have $\s^{k_B}_{Bt}(h^B_t) = \s^{k'_B}_{Bt}(h^B_t)$.
\end{definition} 

In non-revealing equilibria the informed player makes no use of his private information on-path of the equilibrium play. Therefore, the uninformed player can infer nothing from the actions played at each stage, which leaves the ``posterior'' unchanged and equal to the prior. 

In $\mathcal{G}(p^0)$, the number $\text{Cav}(v_A)(p^0_A)+ \text{Cav}(v_B)(p^0_B)$  is an upper bound on the ex-ante equilibrium payoffs of player 1, because each uninformed player can always play the optimal strategy of his repeated zero-sum game, holding the payoffs of the informed player at most at $\text{Cav}(v_A)(p_A)+ \text{Cav}(v_B)(p_B)$.  On the other hand, letting $\mathfrak{h}(p) := v_A(p_A)+ v_B(p_B)$, a lower bound on the ex-ante equilibrium payoffs of the informed player is given by the concavification of $h$ evaluated at $p^0$.\footnote{Consider the two-player infinitely repeated zero-sum game $G_{A+B}(p^0)$.  It is straightforward to check that the non-revealing value of $G_{A+B}(p^0)$ is $\mathfrak{h}(p^0)$. From Aumann et al. \cite{AMS1995}, the value of $G_{A+B}(p^0)$ is $\text{Cav}(\mathfrak{h})(p^0)$, so the informed player can guarantee $\text{Cav}(\mathfrak{h})(p^0)$.}

For $p^0 \in \Delta(K_A \times K_B)$, let $$I(p^0) = [\text{Cav}(\mathfrak{h})(p^0), \text{Cav}(v_A)(p^0_A)+\text{Cav}(v_B)(p^0_B)].$$ We call $\text{Cav}(\mathfrak{h})(p^0)$ the \textit{lower end of} $I(p^0)$ and $\text{Cav}(v_A)(p^0_A)+\text{Cav}(v_B)(p^0_B)$ the \textit{upper end of} $I(p^0)$. The interval $I(p^0)$ might be degenerate as well as non-degenerate. Whenever $I(p^0)$ is non-degenerate, for any sufficiently small perturbation of the stage-game payoff matrices, the resulting model $\mathcal{G}(p^0)$ also has an associated interval $I(p^0)$ which is non-degenerate. A proof of this robustness property can be found in the Supplemental Appendix (Proposition \ref{nondeg}). As an example, the interval $I(p^0)$ of the game $\mathcal{G}(p^0)$ of Example \ref{Example 1} is non-degenerate: we have $\text{Cav}(v_A)(p_A^0) + \text{Cav}(v_B)(p_B^0) = 1/4 + 1 > 1 + 3/16 = \text{Cav}(\mathfrak{h})(p^0)$. There are cases where $I(p^0)$ is degenerate (for example, when $p^0 \in \D(K_A) \bigotimes \D(K_B)$). In these cases only one ex-ante equilibrium payoff exists (without any assumptions on the games $G_A(p^0_A)$ and $G_B(p^0_B))$. The first part of our main result can now be stated: 

\begin{theorem}\label{mainthm1}
Let $p^0 \in \D(K_A \times K_B)$. Suppose there exist non-revealing equilibria in $G_A(p^0_A)$ and $G_B(p^0_B)$. Then the set of ex-ante equilibrium payoffs of the informed player is $I(p^0)$.
\end{theorem}

\subsection*{On the proof of Theorem \ref{mainthm1}} First, the lower end of $I(p^0)$ is always an ex-ante equilibrium payoff for the informed player in $\mathcal{G}(p^0)$ -- i.e., no additional assumption is required. A proof of this result can be found in the Supplemental Appendix (Proposition \ref{existence lower bound}).\footnote{We do not include this proof here because it derives essentially from a long application of the work of Simon et al. \cite{SST1995} to our three-player environment.} When the upper end of $I(p^0)$ is also an ex-ante equilibrium payoff for the informed player, then we can use a straightforward application of \textit{jointly controlled lotteries} developed in Aumann et al. \cite{AMS1995} in order to obtain that the whole interval $I(p^0)$ can be attained as an ex-ante equilibrium payoff to the informed player.\footnote{A jointly controlled lottery is a public randomization device that is endogenously generated by the players so that the players can coordinate. We can define a jointly controlled lottery that randomizes between equilibria paying the upper and lower ends of $I(p^0)$ with any given probability, in order to obtain any given number in $I(p^0)$ as an ex-ante equilibrium payoff of the informed player.  See Aumann et al. \cite{AMS1995} p. 272 for details. For completeness, the explicit construction of the equilibrium involving jointly controlled lotteries is included in the Supplemental Appendix (Proposition \ref{continuum}).} Therefore, the only remaining task to obtain a proof of Theorem \ref{mainthm1} is to show that the upper end of $I(p^0)$ is an ex-ante equilibrium payoff for the informed player.  

We define the set of \textit{non-revealing equilibrium payoffs} (of the informed player), denoted  $\mathcal{NR}(p^0)$, of game $\mathcal{G}(p^0)$. Let $\mathcal{NR}(p^0)$ be the set of vectors $(\phi_A, \phi_B) \in \Re^{|K_A|} \times \Re^{|K_B|}$ that satisfy: 

\begin{enumerate} 
\medskip
\item (Feasibility) $$(\phi_A, \phi_B) \in F_A \times F_B,$$ where $F_A := \text{co} \{ (A^{k_A}_{i_A,j_A})_{k_A \in K_A}| i_A \in I_A, j_A \in J_A\}$ and $F_B := \text{co} \{(B^{k_B}_{i_B,j_B})_{k_B \in K_B}| i_B \in I_B, j_B \in J_B \}.$

\item (Individual rationality for player 1) $$\phi \cdot q \geq \mathfrak{h}(q), \, \forall q \in \D(\supp(p^0)),$$ where  $$\phi = (\phi^{k_A}_A + \phi^{k_B}_B)_{(k_A, k_B) \in \supp(p^0)}, \, \mathfrak{h}(q) = v_A(q_A) + v_B(q_B).$$

\item (Individual rationality for players 2 and 3) $$\phi_A \cdot p^0_A \leq \Cav(v_A)(p^0_A), \, \phi_B \cdot p^0_B \leq \Cav(v_B)(p^0_B).$$

\end{enumerate} 

\medskip

The three conditions above defining the set of non-revealing equilibrium payoffs parallel the conditions defined for two-player repeated games with a single informed player (see Hart \cite{SH1985}). We briefly recall the reason why $(\phi_A, \phi_B) \in \mathcal{NR}(p^0)$ implies we can construct equilibria in $\mathcal{G}(p^0)$ that reveal no information on path of play and have $\phi_A$ (resp. $\phi_B$) as the vectors of payoffs of the informed player in $G_A(p^0_A)$ (resp. $G_B(p^0_B)$). Let $(\phi_A, \phi_B) \in \mathcal{NR}(p^0)$. Since $\phi_A \in F_A$, $$\phi_A \, = \, \sum_{(i_A, j_A) \in I_A \times J_A}\l_{i_A, j_A}(A^{k_A}_{i_A,j_A})_{k_A \in K_A},$$ where  $$\sum_{(i_A, j_A) \in I_A \times J_A}\l_{i_A, j_A} =1, \,  \l_{i_A, j_A}\geq 0.$$ Consider then a sequence of (pure) actions $((i^t_A, j^t_A))_{t \geq 1}$ and define a function $\chi^{(i_A,j_A)}: I_A \times J_A \to \{0,1\}$ where $\chi^{(i_A,j_A)}(i'_A, j'_A) = 1$, if $(i'_A, j'_A) = (i_A, j_A)$, and $\chi^{(i_A,j_A)}(i'_A, j'_A) = 0$, if $(i'_A, j'_A) \neq (i_A, j_A)$. Assume that the sequence $h^A_{\infty} = ((i^t_A, j^t_A))_{t \geq 1}$ satisfies for each $(i_A, j_A) \in I_A \times J_A$, $\lim_{T \to +\infty}\frac{1}{T}\sum^{T}_{t=1} \chi^{(i_A, j_A)}(i^t_A, j^t_A) = \l_{i_A, j_A}$.\footnote{Lemma 2 in Sorin \cite{SS1983} shows such a sequence exists for any such $\l_{i,j}$.} When players 1 and 2 play the sequence of actions $h^A_{\infty}$ in $G_A(p^0_A)$, the payoff achieved is $\phi^{k_A}_A$, for each state $k_A \in K_A$. Obviously, a similar reasoning applies to $\phi_B \in F_B$, and $\phi^{k_B}_B$ is the payoff achieved by a deterministic sequence of actions played at each stage. Conditions (2) and (3) now imply that this deterministic path of play can be supported as an equilibrium path of play: condition (2) guarantees that, in case player 1 deviates from the deterministic sequence, players 2 and 3 can punish him. So it guarantees that under no possible state $(k_A, k_B)$ player 1 could obtain more that $\phi^{(k_A k_B)}$\footnote{Condition (2) implies that there is an optimal strategy of the uninformed player in the game $G_{A+B}(p^0)$ which guarantees the informed player will not obtain more than $\phi^{(k_A,k_B)}$ for each state $(k_A, k_B)$. In the Supplemental Appendix, Proposition \ref{approach strategy}, we show this strategy can indeed be played by players 2 and 3 in $\mathcal{G}(p^0)$.}; condition (3) guarantees, in turn, that any deviation of player $2$ (resp. player $3)$ can be punished by player $1$ with an optimal strategy of game $G_A(p^0_A$)(resp. $G_B(p^0_B)$).

We now define the set $NR_A(p^0_A)$ of non-revealing equilibrium payoffs of the two-player repeated game $G_A(p^0_A)$. The set $NR_A(p^0_A)$ is the set of vectors $\a_A \in \mathbb{R}^{|K_A|}$ that satisfies (i) $\a_A \cdot q \geq v_A(q)$, for all $q \in \D(K_A)$; (ii) $\a_A \cdot p^0_A = \text{Cav}(v_A)(p^0_A)$ and (iii) $\a_A \in F_A$. This is the set of equilibrium payoffs for which no signalling occurs on path. This set is essentially the specification for a two-player, zero-sum infinitely repeated game with one-sided incomplete information of the set of non-revealing equilibrium payoffs (called ``$G$'') defined for nonzero-sum two-player repeated games with lack on information on one side in Hart \cite{SH1985}(see p. $124$). Note that each $k_A$-th entry of a vector $\a_A$ in $NR_A(p^0_A)$ is the payoff for the informed player in $G_A(p^0_A)$ when the state is $k_A$. All vectors $\a_A$ in $NR_A(p^0_A)$ generate the same ex-ante payoff, i.e., the (uniform) value Cav$(v_A)(p^0_A) = \a_A \cdot p^0_A$, but there might be several vectors generating this payoff.  It is now easy to see that $NR_A(p^0_A) \times NR_B(p^0_B) \subseteq \mathcal{NR}(p^0)$. Since we assumed that $NR_A(p^0_A) \times NR_B(p^0_B) \neq \emptyset$, Theorem \ref{mainthm1} now follows immediately. 

\begin{remark}\label{important remark} We remark that property (2) of $\mathcal{NR}(p^0)$ differs from the usual assumption present in the literature, namely, that supp$(p^0) = K_A \times K_B$. The reason why we do not adopt this assumption is as follows: as shown above, $NR_{A}(p^0_A) \times NR_B(p^0_B) \subseteq \mathcal{NR}(p^0)$ for any $p^0 \in \D(K_A \times K_B)$; if $p^0 \in \text{int}(\D(K_A \times K_B))$, then we have in addition that $NR_A(p^0_A) \times NR_B(p^0_B) = \mathcal{NR}(p^0)$: to see this, take $(\phi_A, \phi_B) \in \mathcal{NR}(p^0)$. To show $(\phi_A, \phi_B)$ is in $NR_A(p_A) \times NR_B(p_B)$, we just have to check that property $(i)$ defining the sets $NR_A(p_A)$ and $NR_B(p_B)$ is satisfied; the other conditions are immediate. Suppose by contradiction there exists $\bar q_A \in \D(K_A)$ such that $\phi_A \cdot \bar q_A < v_A(\bar q_A)$. Fix now $\bar q_B \in \D(K_B)$ such that $\phi_B \cdot \bar q_B \leq v_B(\bar q_B)$. It follows that for $\bar q = \bar q_A \bigotimes \bar q_B \in \D(K_A \times K_B)$, $\phi_A \cdot \bar q_A + \phi_B \cdot \bar q_B <  v_A(\bar q_A) + v_B(\bar q_B) = \mathfrak{h}(\bar q)$, which is a contradiction, since $(\phi_A, \phi_B)$ satisfies Condition (2) of $\mathcal{NR}(p^0)$. Hence, when $p^0 \in \text{int}(\D(K_A \times K_B))$, $(\phi_A, \phi_B) \in \mathcal{NR}(p^0)$ implies that $\phi_A$ (respec. $\phi_B$) is a vector of equilibrium payoffs in $G_A(p^0_A)$ (respec. $G_B(p^0_B)$) of the informed player. 

One can see immediately that the argument above relies on the product structure of the set of states $\supp(p^0) = K_A \times K_B$; if one assumes a prior $p^0$ for which $\supp(p^0)$ is not a cartesian product, then the argument above cannot be repeated and, in fact, the claim is not true.\footnote{This reasoning also justifies why the assumption that is adopted throghout is $p^0_A \in \text{int}(\D(K_A))$ and $p^0_B \in \text{int}(\D(K_B))$ and not $p^0 \in \text{int}(\D(K_A \times K_B))$: if, for instance, $p^0_A$ does not have full support, then one can eliminate from the set of states in $K_A$ the ones that have $0$ probability under $p^0_A$ and the resulting space of states is also a product subset of $K_A \times K_B$.} The example below shows that for a certain $p^0 \notin \text{int}(\D(K_A \times K_B))$, $NR_{A}(p^0_A) \times NR_B(p^0_B) \subsetneq \mathcal{NR}(p^0)$. \end{remark}

\begin{example}We follow the notation of Example \ref{Example 1}. For $\e>0$ and $q^0 \in (0,1)$, consider the game $\mathcal{G}(p^0)$ given by the following data:

$$p^0 = \begin{bmatrix} q^0 & 0 \\ 0 & (1-q^0) \end{bmatrix}$$
$$\\
A^1 = \begin{bmatrix} 1 & 0 \\ 0 & 0  \end{bmatrix} \,\, A^2 = \begin{bmatrix} 0 & 0 \\ 0 & 1 \end{bmatrix}
$$
$$
B^1 = \begin{bmatrix} -\e & -\e \\ \e & \e  \end{bmatrix} \,\,B^2 = \begin{bmatrix} \e & \e  \\ -\e & -\e 
\end{bmatrix} 
$$\\

\end{example}

The first thing to observe is that $NR_B(p^0_B) = \emptyset$: notice that $\phi_B \in NR_B(p^0_B)$ iff $\phi_B = (\e, \e)$; but $F_B = \text{co} \{(\e,-\e), (-\e, \e)\}$, which clearly does not contain $(\e,\e)$. In particular, we therefore have that $NR_A(p^0_A) \times NR_B(p^0_B) = \emptyset$. One can now show that $\mathcal{NR}(p^0) \neq \emptyset$. We sketch the proof for completeness: fix $q^0 = \frac{1}{5}$. The game $G_A(p^0_A)$ trivially satisfies the condition $NR$ at $p^0_A$ (see Definition \ref{main def}), because $v_A$ is strictly concave and smooth in the interval $(0,1)$. Therefore, by Theorem \ref{main1}, there exists $\phi_A \in F_A$ such that $\phi_A \cdot p^0_A = v_A(p^0_A)$ and $\phi_A \cdot q \geq v_A(q), \forall q \in \D(K_A)$. For example, one might take $\phi_A = (\frac{16}{25}, \frac{1}{25})$. In $G_B(p^0_B)$ consider now the vector $\phi_B = (-\e, \e)$, which is in $F_B$. Taking $\e>0$ sufficiently small, the vector $\phi \in \Re^{K_A \times K_B}$ whose entries are given by $\phi^{k_A, k_B} := \phi^{k_A}_A + \phi^{k_B}_B, \forall k_A \in K_A, k_B \in K_B$ is in $\mathcal{NR}(p^0)$.

\medskip

\subsection{Second Part of Main Result 1}\label{sec:MR1b}The second part of our main result provides a general sufficient condition for the non-emptyness of $NR_A(p^0_A)$ and $NR_B(p^0_B)$. We introduce a few preliminary definitions in order to state the condition.

Recall that the non-revealing value function $v_A(q)$ is defined by $\text{min}_{\tau}\text{max}_{\s}\s A(q)\t'$, where $\s$ is a row vector and $\t'$ a column vector of the one-shot, two-player, zero-sum game with payoff matrix $A(q)$.  For $q$ in the affine hull $H_A$ of the simplex $\D(K_A)$, one can consider the immediate extension of $v_A(q)$ to $H_A$ given by the same min-max formula. Denote this extension by $v^e_A$. As we will need to make considerations about the derivative of $v^e_A$ at points in $\D(K_A)$, we will define once and for all a parametrization for the affine space $H_A$. Let $T: \Re^{|K_A|-1} \to H_A$ be defined as follows: Let $e^{|K_A|-1}_i = (0,..,1,...,0) \in \mathbb{R}^{|K_A| -1} $ with 1 in the \textit{i}-th position. Analogously, let $e^{|K_A|}_i = (0,..,1,...,0) \in  \Re^{|K_A|}$. Define $T:  \mathbb{R} ^{|K_A| -1} \rightarrow H_A \subseteq \Re^{|K_A|}$ as the affine transformation that maps $e^{K_A-1}_i  \mapsto e^{K_A}_{i+1}$ and $0 \mapsto e^{|K_A|}_{1}$, for $i \in \{1,2,...,|K_A| - 1\}$. Since $T$ is affine, $Tx = Sx + e^{|K_A|}_1$, where $S$ is an injective linear transformation; we will also denote by $S$ the matrix representation of $S$ according to the canonical basis. The function $(v^e_A \circ T): \Re^{|K_A|-1} \to \Re$ is a Lipschitz function and therefore is almost everywhere differentiable in $\Re^{|K_A|-1}$. The \textit{generalized gradient}\footnote{See Clark \cite{FC1975}.} of $(v^e_A \circ T)$ at $x^0$ is defined as $\partial (v^e_A \circ T)(x^0) = \text{co}\{\text{lim} (\nabla (v^e_A \circ T))(x^0 +h_k) | h_k \to 0$ as $k \to +\infty \}$, where $x^0 + h_k \in \R^{|K_A| -1}$ is a point of differentiability of $(v^e_A \circ T)$, for all $k \in \mathbb{N}$.\footnote{Notice that to define the generalized gradient, one needs to take limits from all possible directions in $\Re^{|K_A|-1}$. That is why one needs to extend the non-revealing value function ``outside'' of the simplex.} For notational convenience, we write $\partial v_A(p^0_A) \equiv \partial (v^e_A \circ T)(x^0)$, where $T(x^0) = p^0_A$. Let now $P \subset \Re^{|K_A|-1}$ be such that $T(P) = \D(K_A)$. Define the \textit{restricted superdifferential of} $\Cav(v_A \circ T|_{P})$ at $p$ - denoted $\partial^* \Cav(v_A)(p)$ - as the set of vectors $v \in \Re^{|K_A|-1}$ that satisfy $\Cav(v_A)(p) + v \cdot h \geq \Cav(v_A \circ T|_{P})(x+h)$ for all $h$ with $x+h \in P$ and $T(x) = p$. Below, $\phi_AS$ denotes the pre-multiplication of the row vector $\phi_A$ by the matrix $S$.

\begin{definition}\label{main def}
The two-player infinitely repeated zero-sum game with one-sided incomplete information $G_A(p^0_A)$ satisfies the property \textit{$NR$ at $p^0_A$} if there exists $p_A \in \D(K_A)$ and $\phi_A \in \R^{|K_A|}$ such that: 

\begin{enumerate}
\item $\text{Cav}(v_A)(p_A) = v_A(p_A) = \phi_A \cdot p_A$ and Cav$(v_A)(p^0_A) = \phi_A \cdot p^0_A$;

\item $\phi_A S \in \partial v_A(p_A)$;

\item $\phi_A S \in \partial^* \text{Cav}(v_A)(p_A)$.
\end{enumerate}

\end{definition}
 The properties (1)-(3) in the Definition of $NR$ at $p^0_A$ (Definition \ref{main def}) are properties of $v_A$: (1) states that $p_A$ is a point of identity between $v_A$, $\text{Cav}(v_A)$ and the affine function $q \in \D(K_A) \mapsto \phi_A \cdot q$; (1) also states that $p^0_A$ is a point of identity between the same affine function and $\Cav(v_A)$; (2) states that $\phi_AS$ is a (generalized) gradient at $p_A$ of $v_A$; (3) states that $\phi_A$ is a ``supergradient'' of $\Cav(v_A)$ at $p_A$. We are now ready to state the sufficient condition. 

\begin{theorem}\label{main1}
Let $p^0_A \in \D(K_A)$ such that $\supp(p^0_A) = K_A$. Suppose $G_A(p^0_A)$ satisfies $NR$ at $p^0_A$. Then the set of non-revealing equilibrium payoffs $NR_A(p^0_A)$ is nonempty. Evidently, the analogous statement holds for $G_B(p^0_B)$.
\end{theorem}

A proof of this Theorem can be found in Appendix A. Evidently, for $p^0 \in \D(K_A \times K_B)$ with $\text{supp}(p^0_A) = K_A$ and $\text{supp}(p^0_B) = K_B$, if $G_A(p^0_A)$ satisfies $NR$ at $p^0_A$ and $G_B(p^0_B)$ satisfies $NR$ at $p^0_B$, then it follows from Theorem \ref{mainthm1} that $I(p^0)$ is the ex-ante equilibrium payoff set of the informed player.

\begin{remark} As both definitions of the property $NR$ at $p^0_A$ and the set $NR_A(p^0_A)$ concern the game $G_A(p^0_A)$, we compare their content in detail. If a vector $\phi_A$ satisfies (3) and the first two equalities of (1) of Definition \ref{main def}, then this implies that $\phi_A$ satisfies $(i)$ in the definition of $NR_A(p^0_A)$. If $\phi_A$ satisfies the last equality of (1) of Definition \ref{main def}, then it is immediate that it satisfies $(ii)$ of $NR_A(p^0_A)$. Property $(iii)$ in the definition of $NR_A(p^0_A)$ bears no immediate relation with the non-revealing value function $v_A$. But, as the proof of Theorem \ref{main1} shows, it is implied by a geometric property of $v_A$, namely, property (2). Therefore, information about $NR_A(p^0_A)$ can be infered from properties of $v_A$ only. The precise way to do this inference is presented in the proof of Theorem \ref{main1}. 

Our aim at introducing the condition $NR$ at $p^0_A$ is to highlight that the attainability of the upper end of $I(p^0)$ as an ex-ante equilibrium payoff of the informed player is not purely an ``information problem'', i.e., it does not depend exclusively on the existence of correlation between states in $K_A$ and $K_B$. The property shows how the attainability of the upper end of $I(p^0)$ also fundamentally relies on the payoff structure of the component games, which determines the geometry of the nonrevealing value function. Though the sufficient condition presented might not be easier to check computationally than the direct non-emptyness of $NR_A(p^0_A)$, it isolates the aspects of this geometry which determine the existence of equilibria with such payoffs. 

\end{remark}

\subsection*{An Interpretation for $NR$} Even though the property $NR$ at $p^0_A$ is not straightforward to interpret, as it mainly describes certain geometric properties of the non-revealing value function, there is a class of infinitely repeated two-player zero-sum games with one-sided incomplete information in which this property can be interpreted straightforwardly and in which this property is always satisfied.

In this section, whenever a two-player, zero-sum infinitely repeated game $G_A(p^0_A)$ is considered, it is assumed, without loss of generality, that $p^0_A \in \text{int}(\D(K_A))$. 

\begin{definition}\label{local non reve}
A two-player, infinitely repeated zero-sum game with lack of information on one-side and undiscounted payoffs $G_A(p)$ is \textit{locally non-revealing at $p$}\footnote{We provide a robustness result regarding payoff perturbations for property ``locally non-revealing'' in the Supplemental Appendix (subsection \ref{propnongen}). The result shows that the property is \textit{not} non-generic.}  whenever there exist $k \in \mathbb{N}, (\lambda_i)^{k}_{i=1} \in \mathbb{R}^{k}$ and $(p_i)^{k}_{i =1} \in \prod^{k}_{i=1}\D(K_A)$ such that:
\begin{enumerate}
\item For each $i=1,...,k$, $\l_i >0$ and $\sum^{k}_{i=1}\lambda_i = 1$;
\item $\sum^{k}_{i =1}\lambda_i p_i = p$;
\item $\text{Cav}(v_A)(p)=  \sum^{k}_{i=1}\lambda_i v_A(p_i)$;
\item For some $i_0 \in \{1,...,k\}$, $p_{i_0} \in$ int($\Delta(K_A)$). 
\end{enumerate}
\end{definition}

The definition above implicitly describes an optimal strategy for the informed player in the game $G_A(p)$ for which there is signalling on path of play: the informed player ``splits'' the prior $p$ into finitely many posteriors $(p_i)^k_{i=1}$ such that $\text{Cav}(v_A)(p)=  \sum^{k}_{i=1}\l_i v_A(p_i)$. This is the typical optimal strategy constructed by Aumann et al.. The local non-revelation property tells us that whenever the informed player in $G_A(p^0_A)$ has an Aumman et al.'s strategy under which he does not exclude some state (i.e., when there is at least one induced posterior in the interior of the simplex of states), then there exists an equilibrium of the game $G_A(p^0_A)$ where the informed player does not signal on path of play.  

Figure \ref{FIG:Normals} illustrates how the conditions of Definition \ref{local non reve} are related to the property of $NR$ at $p^0_B$. Let $G_B(p^0_B)$ be the game originating the non-revealing value function $v_B$ (whose graph is depicted in black). The data defining this game is the exact same as in Example \ref{Example 1}. In this figure, let $p^0_B = 1/2$ denote the prior probability of state $1$. One can obviously write $p^0_B$ as a convex combination with equal weights of $1/4$ and $3/4$ (the ``optimal splitting'' which determines the Cav$(v_B)$). Note that at the interior posterior $q = 1/4$ in the figure, the vectors $N_2$ and $N_1$ generate the normal cone to the graph of $v_B$ at $q=1/4$. Each vector of this normal cone is uniquely associated with a supergradient of $v_B$ at $q=1/4$, i.e., $N_m \cdot q = (q, \mathfrak{n}_m \cdot q), m \in \{1,2\}$, so $N_m = (1, \mathfrak{n}_m)$and $\mathfrak{n}_m$ is a super-gradient of $v_A$ at $1/4$. Letting $N = (1, \mathfrak{n})$, note that Cav$(v_B)(q) = \mathfrak{n} \cdot q, q \in [1/4, 3/4]$ and $N$ belongs to that normal cone at $q=1/4$, and therefore $\mathfrak{n}$ can be given by a convex combination of $\mathfrak{n}_1$ and $\mathfrak{n}_2$.

\begin{figure}[t]
\centering{}%
\caption{Local Non-revelation and $NR$ at $p^0_B =1/2$ }
\medskip
\label{FIG:Normals}
\medskip
\includegraphics[scale=0.9]{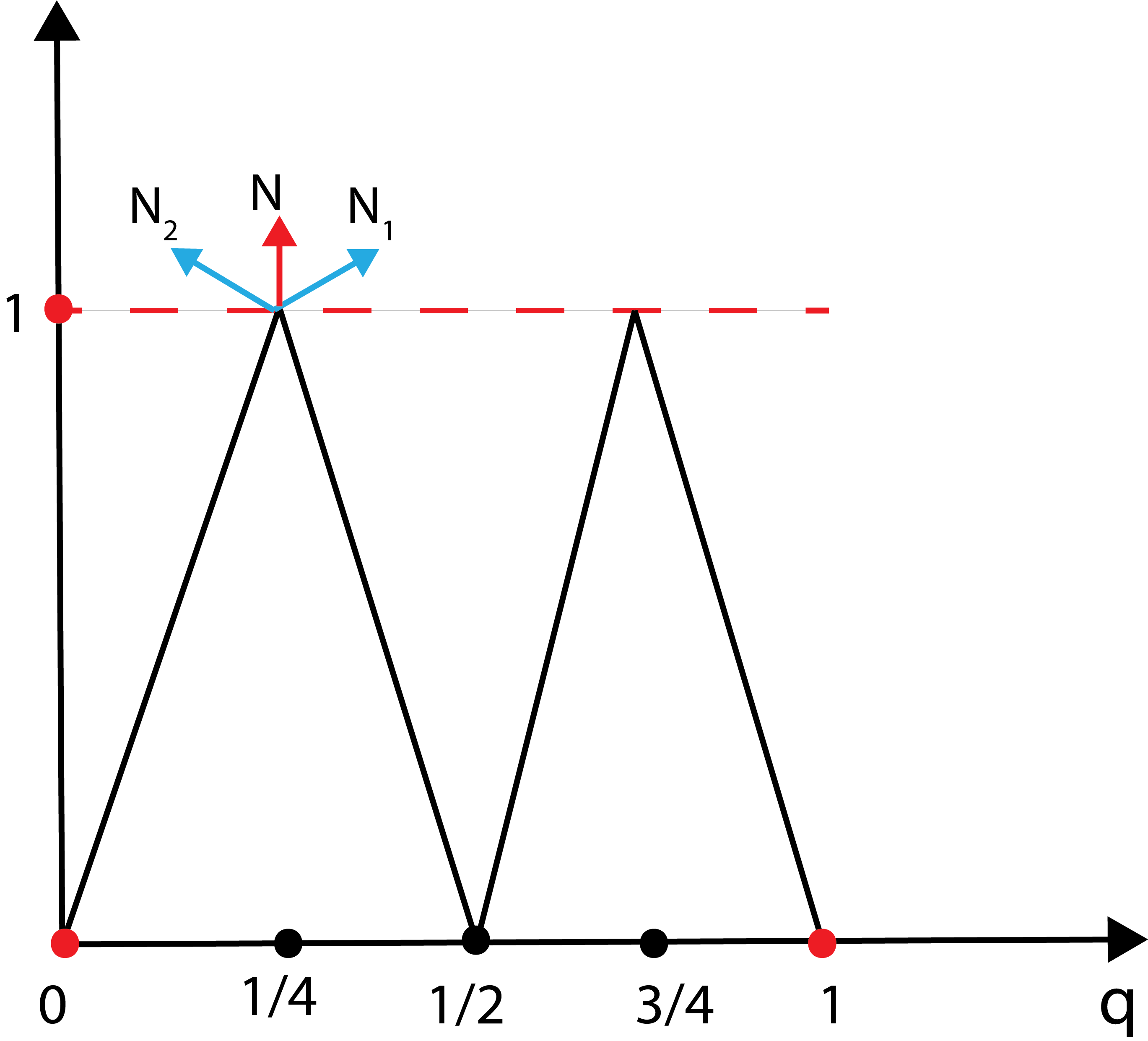}

\end{figure}


\begin{proposition}\label{corgeometric}
Suppose $G_A(p^0_A)$ is locally non-revealing at $p^0_A$. Then  $G_A(p^0_A)$ satisfies $NR$ at $p^0_A$. Evidently, the same statement holds for $G_B(p^0_B)$. \end{proposition}

Proposition \ref{corgeometric} shows that the local non-revelation property implies the $NR$ property. For a proof of this proposition, see Appendix A. Example \ref{attainable} below shows that the $NR$ property is strictly more general than the local non-revelation property. As an example, the games $G_A(p^0_A)$ and $G_B(p^0_B)$ defined in Example \ref{Example 1} are both locally non-revealing at their respective priors. Note that in that example, the game $G_B(p^0_B)$ is such that the optimal strategy of the informed player constructed by Aumann et al. involves signalling on path (this is the signalling strategy we briefly described in the example).

\begin{example}\label{attainable} Let $K_A = \{1,2\}$ be the set of states. Let $q$ denote the probability of state $1$ and $p^0_A = 1/2$ (prior of state 1). Consider the following game $G_A(p^0_A)$:

$$ A^1 = \begin{bmatrix} 0 & 0 \\ 0 & -1 \end{bmatrix}; A^2 = \begin{bmatrix}  -1 & 0 \\ 0 & 0 \end{bmatrix}
$$

\begin{figure}[!h]
\caption{Graphs of $\Cav(v_A)$ and $v_A$}
\begin{tikzpicture}
\medskip
  \draw[->] (0,0) -- (4,0) node[right] {$q$};
  \draw[->] (0,-1) -- (0,2) node[above] {$\Re$};
  \draw[scale=4,domain=0:1,smooth,variable=\p,blue] plot ({\p},{-\p*(1-\p)});
  \draw[scale=4, domain=0:1, densely dotted,variable=\x,red] plot ({\x},{0*(1-\x)});
\node at (4,-0.25) (noderight) {1};
\node at (-0.25,-0.25) (noderight) {0};
\end{tikzpicture}
\end{figure}
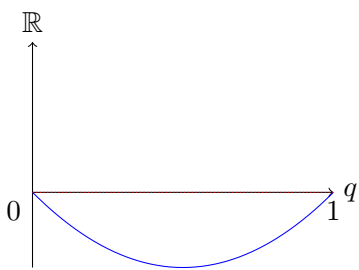

\medskip
For a row vector $v \in \Re^m$, denote by $v'$ the transposed column vector. The figure depicts the graphs of $v_A(q) = -q(1-q)$ and $\Cav(v_A)(q) = 0, \forall q \in [0,1]$. Consider the actions $\s_A = (1,0)$ and $ \t_A = (0,1)'$. Then $\alpha_A = (\s_A A^1 \t_A, \s_A A^2 \t_A) = (0,0)$, so $\alpha^A \cdot q = \text{Cav}(v_A)(q) = 0, \forall q \in [0,1]$. Therefore $\a_A \in NR_A(p^0_A)$. First note that the local non-revelation property at $p^0_A$ is not satisfied in $G_A(p^0_A)$. We show that the property $NR$ at $p^0_A$ is satisfied in the example. The linear transformation $S$ in this example is defined by $Sx = (x, -x)$. Notice that the only candidates for $p$ and $\phi_A$ satisfying the conditions of property $NR$ at $p^0_A$ are $p = 0$ or $p=1$ and $\a_A$. Notice that at $p=0$, just by looking at the graph depicted above, one can see that conditions (1) and (3) of property $NR$ are satisfied for the vector $\a^A$. Now, $\lim_{q \to 0^{-}}\nabla (v^e_A \circ T)(q) =0$, since $v^e_A \circ T$ is constant and equal to $0$ in $(-\infty, 0)$. Using the notation defined above, we have that $\a^A S$ = 0, so that $\a^A S \in \partial v_A(p)$. Observe that because of the strict convexity of the non-revealing value function $v_A$, the optimal strategy of the informed player as constructed in Aumann et al. necessarily involves signalling on path, namely, inducing posteriors at the boundary of the $1$-simplex of states. However, property $NR$ at $p^0_A$ guarantees that an equilibrium exists for which no signalling occurs on path. \end{example}

The next example shows a game $G_A(p^0_A)$ on which $NR$ at $p^0_A$ is not satisfied.

\begin{example}\label{not attainable upper bound} Let $G_A(p^0_A)$ be defined from the payoff matrices below and $q$ denote the probability of state $1$ with prior $p^0_A = 1/2$.

$$ A^1 = \begin{bmatrix} 1 & 1 \\ -1 & -1 \end{bmatrix} \,\, A^2 = \begin{bmatrix}  -1 & -1 \\ 1 & 1 \end{bmatrix}
$$

\begin{figure}[h]
\centering{}%
\caption{Graphs of $\text{Cav}(v_A)$(dotted) and $v_A$(continuous)}\label{fig3}
\bigskip
\begin{tikzpicture}

 \draw[->] (0,0) -- (4,0) node[right] {$q$};
  \draw[->] (0,-1) -- (0,4) node[above] {$\Re$};
  \draw[scale=4,domain=0:0.5,smooth,variable=\p,blue] plot ({\p},{1-2*(\p)});
  \draw[scale=4,domain=0.5:1,smooth,variable=\y,blue] plot ({\y},{-1 + 2*(\y)});
  \draw[scale=4, domain=0:1, densely dotted, variable=\x, red] plot ({\x},{1}); 
  ]

\node at (-0.25,-0.25) (noderight) {0};
\node at (2,-0.25) (noderight) {1/2};
\node at (4,-0.25) (noderight) {1};
\node at (-0.25,4) (noderight) {1};

\end{tikzpicture}
\end{figure}
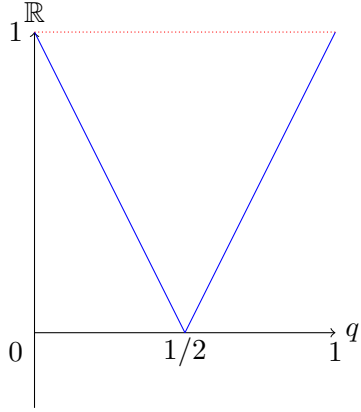
\bigskip
Figure \ref{fig3} depicts the non-revealing value function and its concavification.  The only two candidates for $p_A$ and $\phi_A$ are $p_A=0$ and $p_A=1$ and $(1,1)$. Notice that $v^e_A \circ T$ has a derivative at $p_A=0$, which is equal to $-2$, so $\partial v_A(p_A) = \{-2\}$, whereas $\phi_A S= (1,1) S =0$. Therefore, $\phi_A S \notin \partial v_A(p_A)$. The same reasoning applies to show that $\phi_A S \notin \partial v_A(p_A)$, when $p_A=1$. Therefore, property (2) of Definition \ref{main def} is not satisfied. 
\end{example}

\begin{remark} It is possible to prove directly that if $G_A(p^0_A)$ is locally non-revealing at $p^0_A$ then $NR_A(p^0_A) \neq \emptyset$. The proof is illustrative of how the geometry of $v_A$ determines the existence of equilibria in $G_A(p^0_A)$ that involve no signalling on path of play and therefore we include it here.

The proof is divided in two cases: (a) $\Cav(v_A)(p^0_A) > v_A(p^0_A)$ and (b) $\Cav(v_A)(p^0_A) = v_A(p^0_A)$. We first construct the candidate vector of payoffs $\phi_A \in \Re^{|K_A|}$ for the informed player. Then we show that $\phi_A \in NR_A(p^0_A)$. 

Let $\Cav(v_A)(p^0_A) = \sum_{s \in S}\l_s v_A(p_s)$, $\l_s \geq 0$, $\sum_{s \in S}\l_s =1$  and let $p_{s_0}$ be the posterior which is interior to the simplex of states. Let $\H = \{(q,\a) \in \D(K_A) \times \Re | \Cav(v_A)(q) \geq \a \}$. Note that $\H$ is a convex subset of $\Re^{|K_A|} \times \Re$. Now, note that $(p^0_A, \Cav(v_A)(p^0_A)) = \sum_{s \in S}\l_s(p_s, v_A(p_s))$. Since, by (a), $\forall s \in S, p_s \neq p_0$, it follows that $(p^0_A, \Cav(v_A)(p^0_A))$ is not an extremum point of $\H$ (cf. \cite{TR1970}, Section 18). Therefore, there exists a face $F$ of $\H$, with dimension $d \geq 1$, such that $F$ contains $(p^0_A, \Cav(v_A)(p^0_A))$ and each point $(p_s, v_A(p_s))$. Take now a supporting hyperplane $H$ to $\H$ which contains $F$. The hyperplane $H$ intersected with $\D(K_A) \times \Re$ is the graph of an affine function $q \in \D(K_A) \mapsto \phi_A \cdot q \in \Re$, where $\phi_A$ is a vector in $\Re^{|K_A|}$. We now claim the following:


\begin{claim} The vector $\phi_A$ belongs to $NR_A(p^0_A)$. \end{claim} 

\begin{proof} As the graph of the affine function $q \in \D(K_A) \mapsto \phi_A \cdot q$ is $H \cap (\D(K_A) \times \Re)$, it follows that $\phi_A \cdot q = \Cav(v_A)(q), \forall q \in [p_{s_0}, p^0_A]$, where $[p_{s_0}, p^0_A]$ denotes the segment between $p_{s_0}$ and $p^0_A$. In particular, $\Cav(v_A)(p^0_A) = \phi_A \cdot p^0_A$ is satisfied, and so $\phi_A$ satisfies (ii). From the fact that $H$ supports $\H$ at $F$, it follows that $\phi_A \cdot q \geq v_A(q), \forall q \in \D(K_A)$. This proves $\phi_A$ satisfies $(i)$. Thus it only remains to show $\phi_A \in F_A$. Let $\s_A$ be the Aumann et al. optimal strategy of the informed player in $G_A(p_{s_0})$, which is state-independent since $\Cav(v_A)(p_{s_0}) = v_A(p_{s_0})$.  For the uninformed player, there exists a Blackwell strategy \footnote{Cf. the next subsection for a discussion and definition of Blackwell strategies.} $\t_A$ and a constant $L>0$ such that $\mathbb{E}^{k_A}_{\s_A,\tau_A, p_{s_0}}[\frac{1}{T} \sum^{T}_{t=1}(A^{k_A}_{i^t_A, j^t_A})] \leq \phi^{k_A}_A + \frac{L}{\sqrt{T}}, \forall k_A \in K_A$ (see proof of Corollary 3.34 in Sorin \cite{SS2002}), where $\mathbb{E}^{k_A}_{\s_A,\tau_A, p_{s_0}}$ is the conditional expectation given $\k_A = k_A$. The profile $(\s_A, \t_A)$ is therefore a uniform equilibrium of $G_A(p_{s_0})$, so it follows that the limit $\lim_{T \to +\infty}\mathbb{E}^{k_A}_{\s_A,\tau_A, p_{s_0}}[\frac{1}{T} \sum^{T}_{t=1}(A^{k_A}_{i^t_A, j^t_A})]$ exists for each $k_A \in K_A$.

For each $k_A \in K_A$, let $v^{k_A}$ be this limit. It follows from the inequality of the Blackwell strategy given above that $\forall k_A \in K_A, v^{k_A} \leq \phi^{k_A}_A$. Since for each $T \geq 1$, $$\mathbb{E}_{\s_A,\tau_A, p_{s_0}}\Big[\frac{1}{T} \sum^{T}_{t=1}(A^{\kappa_A}_{i^t_A, j^t_A})\Big] = \sum_{k_A \in K_A} p^{k_A}_{s_0} \Big(\mathbb{E}^{k_A}_{\s_A,\tau_A, p_{s_0}}\Big[\frac{1}{T} \sum^{T}_{t=1}(A^{k_A}_{i^t_A, j^t_A})\Big]\Big)$$ it follows that $\sum_{k_A \in K_A} p^{k_A}_{s_0}v^{k_A} = \Cav(v_A)(p_{s_0})$, as any uniform equilibrium pays $\Cav(v_A)(p_{s_0})$. Since $p_{s_0} \in \text{int}(\D(K_A))$, it must be that $v^{k_A} = \phi^{k_A}_A$, for all $k_A \in K_A$. Since $F_A$ is compact and convex and neither $\s_A$ nor $\t_A$ condition on realized states, it follows that $(v^{k_A})_{k_A \in K_A} = v \in F_A$. This shows $\phi_A \in F_A$ and concludes that $\phi_A \in NR_A(p^0_A)$, as claimed. 

The proof of case (b) is immediate, since the Aumann et al. optimal strategy of the informed player in $G_A(p^0_A)$ is already state-independent. This concludes the proof of the claim. \end{proof}
\medskip

We would like to highlight two things about the above remark. Notice that the first part of the proof does not rely on any assumption on $G_A(p^0_A)$ being locally non-revealing at $p^0_A$, i.e., it does not rely on the assumption that an optimal strategy of the informed player exists inducing an interior posterior. The first paragraph just serves the purpose of pinning down the correct vector $\phi_A$. Notice that the assumption of $p_{s_0} \in \text{int}(\D(K_A))$ is used in the proof above only when we show that $\phi_A = v$. Second, we did not show in the proof that the profile $(\s_A, \t_A)$ is an equilibrium of $G_A(p^0_A)$; we used the fact that $(\s_A,\t_A)$ is a uniform equilibrium of $G_A(p_{s_0})$, inducing the vector of payoffs for the informed player equal to $\phi_A$, and obtained that $\phi_A \in NR_A(p^0_A)$. \end{remark}

\subsection{Consequences for Two-player Zero-sum Repeated Games with Incomplete Information}\label{SEC: 2-player games} Theorem \ref{main1} tells us that given a two-player zero-sum game with one-sided incomplete information that satisfies the property $NR$ at the prior, there exists an equilibrium of the 2-player game for which no signalling occurs on path of play (i.e., the set of non-revealing equilibrium payoffs of this game is non-empty). Even for games for which Aumann et al. constructed strategies that necessarily involved signalling on path (see Example \ref{attainable} or  $G_B(p^0_B)$ in Example \ref{Example 1}), the $NR$ property implies the existence of an equilibrium for which no signalling occurs on path.  We would like to illustrate this message with an example. 

Consider the two-player, zero-sum infinitely repeated game with one-sided incomplete information $G_B(1/2)$ between players 1 and 3, defined by the data in Example \ref{Example 1}. Following Aumann et al.'s  technique for constructing optimal strategies, the strategy of the informed player would be the strategy highlighted in Example \ref{Example 1}, that is, a \textit{signalling strategy}: the informed player uses his actions to signal information about the underlying state for finitely many stages and after that plays the (mixed) optimal action of the one-shot, zero-sum game given by the posterior at each state independently. Playing this strategy in $G_{B}(p^0_B)$ guarantees to the informed player an ex-ante payoff of $(1/2) v_B(1/4) + (1/2) v_B (3/4) = (1/2) \text{Cav}(v_B)(1/4) + (1/2) \text{Cav}(v_B)(3/4) = \text{Cav}(v_B)(1/2) = 1$. Now, the optimal strategy for the uninformed player in $G_B(1/2)$ is a so-called \textit{approachability strategy} or \textit{Blackwell strategy}. For a general game $G_B(p^0_B)$, a Blackwell strategy can be defined as follows: given $\phi_B \in \Re^{|K_B|}$ s.t. $\phi_B \cdot q \geq v_B(q), q \in \D(K_B), \Cav(v_B)(p^0_B) = \phi_B \cdot p^0_B$, $\t_B$ is a \textit{Blackwell strategy} (for $\phi_B$ and $p^0_B$), if for each $\e>0$, there exists $T_0 \in \mathbb{N}$ such that $\forall T \geq T_0$, $\s_B$ a strategy of the informed player and $k_B \in K_B$, 

$$\mathbb{E}^{k_B}_{\sigma_B,\tau_B,p^0}\Big[\frac{1}{T} \sum^{T}_{t=1}(B^{\kappa_B}_{i^t_B, j^t_B})\Big] \leq \phi^{k_B}_B + \e,$$ where  $\mathbb{E}^{k_B}_{\sigma_B,\tau_B,p^0}$ is the conditional expectation given $\kappa_B = k_B$. So $\t_B$ precludes the informed player from achieving more than $\phi^{k_B}_B + \e$ in a sufficiently long (but finitely repeated) game, for any state $k_B \in K_B$. This implies that the ex-ante expected payoff to the informed player in a sufficiently long game is not larger than $\Cav(v_B)(p^0_B) + \e$. For our example, take $\phi_B = (1,1)$ and consider $\t_B$ the Blackwell strategy for $\phi_B$ and $p^0_B$ as defined in the example. It follows the pair $(\s_B, \t_B)$ is a (uniform) equilibrium with associated payoff $\Cav(v_B)(p^0_B)=1$. 

We call the strategies just defined \textit{standard optimal strategies}. In contrast to these standard optimal strategies, for which there is revelation of information on path of play, we now construct equilibrium strategies for both players for which \textit{no information is revealed on path}. The idea for the construction of these strategies is simple. Both players play a deterministic sequence of actions as long as the other plays it. This deterministic sequence is supported by punishment strategies, in the sense that if any player deviates from his sequence of actions, the other player starts to play the punishment strategy forever. Let us first define the deterministic path of play for this example: the deterministic sequence of actions is defined by $((i^t_B, j^t_B))^{\infty}_{t=1} = ((U,R), (U,M), (U,R), (U,M),...)$; so the uninformed player alternates between $R$ (right column) and $M$ (middle column), whereas the informed player plays $U$ at every stage. For each state $k_B \in \{1,2\}$, the payoff (to the informed player) obtained from this path of play is: $$\lim_{T \to +\infty}\frac{1}{T}\sum^{T}_{t=1}B^{k_B}_{i^t_B, j^t_B} = 1.$$This implies the ex-ante payoff to the informed player is also $1 = \text{Cav}(v_B)(1/2)$. In case any player deviates from his prescribed sequence of actions, the other player can observe the deviation and play from the next stage onwards his standard optimal strategy in $G_B(1/2)$, which guarantees that a deviation is not profitable. As mentioned, the strategies just defined are also a (uniform) equilibrium in $G_B(1/2)$, but reveal no information on path of play.

\section{Main Result 2: Necessary Condition for the Upper End of $I(p^0)$ to be Attained in Equilibrium}
\label{SEC: empty}
\medskip  

Theorem \ref{nec condition} is the main result of this section. Intuitively, it shows that equilibria which pay the upper end of $I(p^0)$ to the informed player imply a particular type of signalling procedure. Theorem \ref{nec condition} can therefore be viewed as a constraint to the signalling strategies an informed player might play in an equilibrium paying the upper end of $I(p^0)$. Presenting the formal statement of the Theorem requires some preliminary definitions.

Endow $I_A, I_B, J_A, J_B, K_A$ and $K_B$ with the discrete topology. For each $t \in \mathbb{N}$, let $\mathcal{H}_t$ be the discrete field over $H_{t} := (I_A \times I_B \times J_A \times J_B)^{t-1}$. Endow $H_{\infty} := (I_A \times I_B \times J_A \times J_B)^{\infty}$ and $\Omega := H_{\infty} \times K_A \times K_B$ with the induced product topology and let $\mathcal{H}_{\infty}$ and $\mathcal{F}(\O)$ be the Borel sigma-fields over $H_{\infty}$ and $\O$, respectively.  For notational convenience, we will also denote by $\mathcal{H}_t$ the field generated by $H_t$ on $\O$. 

Let $(\sigma, \tau_A, \tau_B)$ be a profile of strategies in $\mathcal{G}(p^0)$. Let $\mathbb{P}_{\s, \t_A, \t_B, p^0}$ be the probability induced by $(\s, \t_A, \t_B, p^0)$ on $(\O, \mathcal{F}(\O))$. We define the \textit{martingale of posteriors} obtained through Bayesian updating. For $t \in \mathbb{N}$,  let $p^{k_A,k_B}_t : =\mathbb{P}_{\sigma,\tau_A,\tau_B,p}(\kappa = (k_A, k_B)| \mathcal{H}_t)$ and $p_t : = (p^{k_A, k_B}_t)_{(k_A,k_B) \in K_A \times K_B}$. The sequence $(p_t)_{t \in \mathbb{N}}$ is a $\Delta(K_A \times K_B)$-valued martingale with respect to $(\mathcal{H}_t)_{t \in \mathbb{N}}$, satisfying: (i) $p_1 = p^0$; (ii) there exists $p^{\infty}$ such that $p_t \rightarrow p^{\infty}$ a.s. as $t \rightarrow + \infty$. The a.s. limit $p^{\infty}$ of the process $(p_t)_{t \in \mathbb{N}}$ is called the \textit{asymptotic posterior}.

\begin{theorem}\label{nec condition}
Let $(\s, \t_A, \t_B)$ be an equilibrium of $\mathcal{G}(p^0)$ and let $(p_{t})_{t \in \mathbb{N}}$ be the martingale of posteriors induced by the equilibrium. Assume $(\s, \t_A, \t_B)$ pays ex-ante $\Cav(v_A)(p^0_A) + \Cav(v_B)(p^0_B)$ to the informed player. Then $(\Cav(v_{\ell})(p_{t\ell}))_{t \in \mathbb{N}}$ is a martingale, for each $\ell \in \{A,B\}$. 
\end{theorem} 

For an intuition on Theorem \ref{nec condition}, note that from Jensen's inequality, it is immediate to see that Cav$(v_A)(p_{t \ell})$ is a supermartingale w.r.t. $(\H_t)_{t \in \mathbb{N}}$. If Cav$(v_A)(p_{t \ell_0})$ is not a martingale for some $\ell_0 \in \{A,B\}$, then this would imply that the expected payoffs induced by the equilibrium in game $G_{\ell_0}(p^0_{\ell_0})$ are less than Cav$(v_A)(p^0_A)$, contradicting the assumption of the Theorem. 

The next corollary of Theorem \ref{nec condition} is motivated by the following intuition. Take a game $\mathcal{G}(p^0)$ for which $p^0 \in \text{int}(\D(K_A \times K_B))$ and assume that $\mathcal{NR}(p^0) = \emptyset$. This last assumption implies that if the upper end of $I(p^0)$ can be attained in equilibrium in $\mathcal{G}(p^0)$, then the informed player \textit{must} use a state-dependent strategy on path, i.e., he must use signalling. An idea of how such signalling procedure could occur is as follows: the informed player could signal so as to induce posteriors that are product distributions;  this would imply that, once the posterior realizes, from then onwards no correlation exists between the zero-sum games at the posteriors and therefore the informed player could play each of the zero-sum games without concern for information spillover; is it possible that there are equilibria of $\mathcal{G}(p^0)$ for which the described signalling procedure exists and the upper end of $I(p^0)$ can be attained in equilibrium? Corollary \ref{product} provides an answer to this question.  

\begin{corollary}\label{product}Let $(\s,\t_A,\t_B)$ be an equilibrium of $\mathcal{G}(p^0)$. Suppose this equilibrium pays the upper end of $I(p^0)$ and its associated asymptotic posterior $p^{\infty}$ is a product\footnote{By $p^{\infty}$ being a \textit{product a.s.} we mean that there exists $X: \O \to \D(K_A)$ and $Y: \O \to \D(K_B)$, both $\mathcal{F}(\O)$-measurable, such that $p^{\infty} = X \otimes Y$ a.s..} a.s.. Then $I(p^0)$ is degenerate.\end{corollary}

\begin{proof}We will denote by $\mathbb{E}$ the expectation operator $\mathbb{E}_{\s, \t_A, \t_B, p^0}$. Since $p^{\infty}$ is a product a.s., then $\Cav(\mathfrak{h})(p^{\infty}) = \Cav(v_A)(p^{\infty}_A) + \Cav(v_B)(p^{\infty}_B)$ a.s.. Taking expectations on both sides and using Jensen's inequality, it follows that $\Cav(\mathfrak{h})(p^{0}) \geq \mathbb{E}[\Cav(v_A)(p^{\infty}_A)] + \mathbb{E}[\Cav(v_B)(p^{\infty}_B)]$. By Theorem \ref{nec condition}, $\mathbb{E}[\Cav(v_A)(p^{\infty}_A)] = \Cav(v_A)(p^0_A)$ and $\mathbb{E}[\Cav(v_B)(p^{\infty}_B)] = \Cav(v_B)(p^0_B)$.  Therefore, $\Cav(\mathfrak{h})(p^{0}) \geq \Cav(v_A)(p^0_A) + \Cav(v_B)(p^0_B)$. Since, $\Cav(\mathfrak{h})(p^{0}) \leq \Cav(v_A)(p^0_A) + \Cav(v_B)(p^0_B)$, we have that $I(p^0)$ is degenerate. \end{proof}

Therefore, if $I(p^0)$ is non-degenerate, no equilibrium paying the upper end of $I(p^0)$ to the informed player induces a product asymptotic posterior. In other words, such an equilibrium must maintain the correlation (even at infinity) between the zero-sum games with positive probability.

\subsection{On the Proof of Theorem \ref{nec condition}} The proof of Theorem \ref{nec condition} requires some preliminary work. In particular, it requires an auxiliary Lemma (Lemma \ref{characterization}) which provides a necessary condition for equilibria of $\mathcal{G}(p^0)$.

\begin{lemma}\label{characterization}
Let $(\s, \t_A, \t_B)$ be an equilibrium in $\mathcal{G}(p^0)$. Then there exists a sequence of random variables $(p_t, \b_{A,t}, \b_{B,t})_{t \in \mathbb{N}}$ on the probability space $(\O, \mathcal{F}(\O), \mathbb{P}_{\s,\t_A,\t_B,p^0})$  taking values in $\D(K_A \times K_B) \times \Re \times \Re$ such that:

\begin{enumerate}

\item  $(p_t, \b_{A,t}, \b_{B,t})_{t \in \mathbb{N}}$ is a martingale adapted to $(\mathcal{H}_t)_{t \in \mathbb{N}}$.

\item $\b_{A,1} + \b_{B,1}$ is the expected payoff of the equilibrium to player 1. 

\item $\b_{A,t} \leq \Cav(v_A)(p_{tA})$ a.s., $\forall t \in \mathbb{N}$. 

\item $\b_{B,t} \leq \Cav(v_B)(p_{tB})$ a.s., $\forall t \in \mathbb{N}$.

\end{enumerate}

\end{lemma}

\begin{proof}[Proof of Theorem \ref{nec condition}] First, notice that for each $\ell \in \{A,B\}$ and $k,s \in \mathbb{N}$ with $k \leq s$ we have that $\text{Cav}(v_{\ell})(p_{k\ell}) = \text{Cav}(v_{\ell})(\mathbb{E}[p_{s \ell}| \mathcal{H}_k]) \geq \mathbb{E}[\text{Cav}(v_{\ell})(p_{s \ell}) | \mathcal{H}_k]$ a.s. -- where the equality follows from the fact that $(p_s)_{s \in \mathbb{N}}$ is a martingale, and the inequality follows from Jensen's inequality. Assume by contradiction that there exist $k,s \in \mathbb{N}$ with $k < s$,  $\ell_0 \in \{A,B\}$ and an atom $h_k \in H_{k}$ such that Cav$(v_{\ell_0})(p_{k\ell_0})(h_k) > \mathbb{E}[\text{Cav}(v_{\ell_0})(p_{s \ell_0})| \mathcal{H}_{k}](h_k)$. It follows that $\text{Cav}(v_{\ell_0})(p^0_{\ell_0}) \geq \mathbb{E}[\text{Cav}(v_{\ell_0})(p_{k\ell_0})] > \mathbb{E}[\text{Cav}(v_{\ell_0})(p_{s\ell_0})] \geq \mathbb{E}[\beta_{\ell_0,s}] = \beta_{\ell_0,1}$, where the first inequality is given by Jensen's inequality, the second by assumption, the third by (3) and (4) of Lemma \ref{characterization} and the last equality by the martingale property. This then implies that $\text{Cav}(v_A)(p^0_A) + \text{Cav}(v_B)(p^0_B) > \beta_{A,1} + \beta_{B, 1}$. Contradiction, since by (2) of Lemma \ref{characterization} we have that $\beta_{A,1} + \beta_{B,1} =  \text{Cav}(v_A)(p^0_A) + \text{Cav}(v_B)(p^0_B)$. \end{proof}

\begin{example}\label{not attainable} In the next example, $I(p^0)$ is non-degenerate and the upper end of $I(p^0)$ is not an ex-ante equilibrium payoff. More precisely, only the lower end of $I(p^0)$ is an ex-ante equilibrium payoff for the informed player. 
We will provide a proof of this claim through an application of Theorem \ref{nec condition}. Later in a remark,  we provide a more elementary proof of this claim, which will not make any reference to the stochastic process of payoffs and posteriors of Lemma \ref{characterization}. Consider $\mathcal{G}(p^0)$ defined by the following data:

$$ p^0 = \begin{bmatrix} 1/2  & 0 \\ 0 & 1/2 \end{bmatrix}$$

$$ A^1 = \begin{bmatrix} 1 & 1 \\ -1 & -1 \end{bmatrix} \,\, A^2 = \begin{bmatrix}  -1 & -1 \\ 1 & 1 \end{bmatrix}
$$

$$ 
B^1 = \begin{bmatrix} 1 & 0 \\ 0 & 0 \end{bmatrix} \,\, B^2 = \begin{bmatrix}  0 & 0 \\ 0 & 1 \end{bmatrix}
$$

\medskip

\begin{claim} For the game $\mathcal{G}(p^0)$ defined by the data above, only the lower end of $I(p^0)$ is an ex-ante equilibrium payoff for the informed player. \end{claim}

\begin{proof}Assume by way of contradiction that $(\sigma, \tau_A, \tau_B)$ is an equilibrium that pays ex-ante $\text{Cav}(v_A)(p_A^0)+\text{Cav}(v_B)(p_B^0)$ for the informed player in $\mathcal{G}(p^0)$. Let $V_A : \triangle(K_A)\rightarrow\mathbb{R}$ be given by  $V_A(p) := \text{max}_{\sigma ,\tau}\{ \sigma A(p) \tau | \sigma A(p) \tau \leq \text{Cav}(v_A)(p) \}$, i.e., the maximum payoff player 1 attains in the one-shot zero-sum game with payoff matrix $A(p)$ which is less than \text{Cav}$(v_A)(p)$. For this example we have that  $V_A(p) = v_A(p),  \forall p \in \triangle(K_A)$, which can be checked by computation. Let $(\beta_{A,s})_{s \in \mathbb{N}}, (\beta_{B,s})_{s \in \mathbb{N}}$ be the martingales from Lemma \ref{characterization}. We need the following auxiliary claim, whose proof is left to Appendix B. 

\begin{claim}\label{aux claim}
For each $t \in \mathbb{N}$, $\beta_{A,t} \leq V_A(p_{tA}) + Z_t$ a.s., where $(Z_t)_{t \in \mathbb{N}}$ is a (a.s.) nonnegative, bounded sequence that converges (a.s.) to $0$.\end{claim} 

By the claim, we have $\beta_{A,t} \leq v_A(p_{tA}) + Z_t$ a.s.. The Martingale Convergence Theorem now implies that  $\beta_{A,\infty} \rightarrow \beta_{A, \infty}, p_{tA} \rightarrow p^{\infty}_{A}$, as $t \to +\infty$. By the claim $Z_t \to 0$ a.s.. Therefore, we obtain $\mathbb{E}[\beta_{A,\infty}] \leq \mathbb{E}[v_A(p^{\infty}_A)]$. From $(2)$ and $(3)$ in Lemma \ref{characterization}, we have that $\mathbb{E}[\beta_{A,\infty}] = \text{Cav}(v_A)(p^0_A) \geq \mathbb{E}[v_A(p^{\infty}_A)]$. So it follows that $\mathbb{E}[\beta_{A,\infty}] = \mathbb{E}[v_A(p^{\infty}_{A})]$, which implies that the distribution of $p^{\infty}_{A}$ is concentrated at the boundary of $\Delta(K_A)$. Since $v_B$ is strictly concave\footnote{See Example 1 for the formula of $v_B$ and depiction of its graph.} and $(v_B(p_{tB}))_{t \in \mathbb{N}}$ is a martingale (by Theorem \ref{nec condition}), it follows that $p_{tB} = p^0_B$ a.s., $\forall t \in \mathbb{N}$. Hence, we have that for any history $h_{\infty}$ outside a set of $\mathbb{P}_{\sigma,\tau_A, \tau_B,p}$-measure zero, the matrix representation of $p^{\infty}(h_{\infty})$ has either the first or the second row filled with zeros (recall that an entry $p^0_{ij}$ represents the probability of states $i$ and $j$ in game $G_A(p^0_A)$ and $G_B(p^0_B)$, respectively), i.e., $p^{\infty}(h_{\infty})$ is either: 

\begin{center}
$$
 \begin{bmatrix} 1/2 & 1/2 \\ 0 & 0 \end{bmatrix} \text{ or }\begin{bmatrix} 0 & 0 \\ 1/2 & 1/2 \end{bmatrix} 
$$

\end{center}

\bigskip 

Now the process of posteriors is a martingale, which implies that the expectation of $p^{\infty}$ is $p^0$. This implies that the following equation has a solution in $\lambda \in [0,1]$:  

\bigskip 

$$
\begin{bmatrix} 1/2  & 0 \\ 0 & 1/2  \end{bmatrix} =  \lambda \begin{bmatrix} 0 & 0 \\ 1/2 & 1/2 \end{bmatrix} + (1 - \lambda) \begin{bmatrix} 1/2 & 1/2 \\ 0 & 0 \end{bmatrix}.
$$

\bigskip

But this equation has no solution for $\lambda \in [0,1]$, which finally implies a contradiction. Hence, there is no equilibrium paying ex-ante to the informed player $\text{Cav}(v_A)(p^0_A)+\text{Cav}(v_B)(p^0_B)$. The arguments above give us more: recall that we had $\beta_{A, \infty} \leq V_A(p^{\infty}_A) = v_A(p^{\infty}_A) $ a.s. and since $\beta_{B, \infty} \leq \text{Cav}(v_B)(p^{\infty}_B) = v_B(p^{\infty}_B)$ a.s., these imply that $\beta_{A, \infty} + \beta_{B,\infty} \leq v_A(p^{\infty}_A) + v_B(p^{\infty}_B)$ a.s. and therefore $\mathbb{E}[\beta_{A, \infty} + \beta_{B, \infty}] \leq \mathbb{E}[v_A(p^{\infty}_A) + v_B(p^{\infty}_ B)] \leq  \mathbb{E}[\text{Cav}(\mathfrak{h})(p^{\infty})] \leq \text{Cav}(\mathfrak{h})(p^0)$, where the second inequality follows by definition of $\text{Cav}(\mathfrak{h})$ and the last inequality is given by Jensen's inequality. The number $\text{Cav}(\mathfrak{h})(p^0)$ is the lowest possible ex-ante equilibrium payoff to the informed player. This implies that every uniform equilibrium of the example pays $\text{Cav}(\mathfrak{h})(p^0)$ to the informed player. \end{proof}

\begin{remark}We would like to provide an alternative proof of the claim that in the game of Example \ref{not attainable} only the lower end of $I(p^0)$ is an equilibrium payoff.\footnote{We thank an anonymous referee for the suggestion of this alternative proof.} It is obvious that the matrix $A^1$ (respec. $A^2$) can be substituted by the following equivalent matrix $A^1_r$ (respec. $A^2_r$), by simply eliminating the redundant column action of player 2. So,

$$ p^0 = \begin{bmatrix} 1/2  & 0 \\ 0 & 1/2 \end{bmatrix}$$

$$ A^1_r = \begin{bmatrix} 1 \\ -1 \end{bmatrix} \,\, A^2_r = \begin{bmatrix}  -1 \\ 1 \end{bmatrix}
$$

$$ 
B^1 = \begin{bmatrix} 1 & 0 \\ 0 & 0 \end{bmatrix} \,\, B^2 = \begin{bmatrix}  0 & 0 \\ 0 & 1 \end{bmatrix}
$$

The model $\mathcal{G}(p^0)$ defined by the data above is evidently equivalent to the model $\mathcal{G}(p^0)$ defined by the data of the previous example, as only redundant actions have been eliminated, which leaves the best-reply correspondences of all players unaltered. Player 2 now, evidently, is a dummy player. We can define therefore a two-player non-zero-sum game between players $1$ and $3$, whose equilibria immediately induce the equilibria of the three-player game $\mathcal{G}(p^0)$. To be precise, we define a two-player, non-zero-sum infinitely repeated game with one-sided incomplete information and undiscounted payoffs $\mathbb{G}(q^0)$, where the set of states will be $K=\{1,2\}$ with prior $q^0=1/2$ for state $1$:  the payoffs are given by the following bimatrix $C^k$ ($k \in K$), where the informed player plays row and the uninformed player 3 plays column:


 $$C^1: \begin{array}{ccc}
\multicolumn{1}{c}{} &\multicolumn{1}{c}{L} &\multicolumn{1}{c}{R} \\
\cline{2-3}
\multicolumn{1}{c}{(U,U)} &\multicolumn{1}{|c}{(2,-1)} &
\multicolumn{1}{|c|}{(1,0)} \\
\cline{2-3}
\multicolumn{1}{c}{(U,D)} &\multicolumn{1}{|c}{(1,0)} &
\multicolumn{1}{|c|}{(1,0)}\\
\cline{2-3}
\multicolumn{1}{c}{(D,U)} &\multicolumn{1}{|c}{(0,-1)} &
\multicolumn{1}{|c|}{(-1,0)}\\
\cline{2-3}
\multicolumn{1}{c}{(D,D)} &\multicolumn{1}{|c}{(-1,0)} &
\multicolumn{1}{|c|}{(-1,0)}\\
\cline{2-3}
\end{array}\qquad
C^2:\; \begin{array}{cccc}
\multicolumn{1}{c}{} &\multicolumn{1}{c}{L} &\multicolumn{1}{c}{R} \\
\cline{2-3}
\multicolumn{1}{c}{(U,U)} &\multicolumn{1}{|c}{(-1,0)} &
\multicolumn{1}{|c|}{(-1,0)} \\
\cline{2-3}
\multicolumn{1}{c}{(U,D)} &\multicolumn{1}{|c}{(-1,0)} &
\multicolumn{1}{|c|}{(0,-1)}\\
\cline{2-3}
\multicolumn{1}{c}{(D,U)} &\multicolumn{1}{|c}{(1,0)} &
\multicolumn{1}{|c|}{(1,0)}\\
\cline{2-3}
\multicolumn{1}{c}{(D,D)} &\multicolumn{1}{|c}{(1,0)} &
\multicolumn{1}{|c|}{(2,-1)}\\
\cline{2-3}
\end{array}
$$
\medskip

The rows and column labels in matrix $C^k$ should be read as follows: $L$ and $R$ stand for the stage-game actions of the uninformed player (i.e., player 3). For the row player (i.e., player 1), $(U,D)$ corresponds to choosing the top row in game in $A^{k}$ and the bottom row in game $B^{k}$. A generic entry is therefore $C^k_{(i_A, i_B), j_B} \equiv ((A^{k}_r)_{i_A} + B^{k}_{i_B, j_B}, -B^k_{i_A, j_B})$. The other entries are analogously constructed. 

We can now modify the stage-game payoffs $C^k, k=1,2$, so that the best-reply correspondence of both players remains unaltered and, after the modification, we obtain a zero-sum game between players $1$ and $3$. Define new stage-game payoff matrices $D^k, k=1,2$ by: $D^k_{(i_A, i_B) , j_B} \equiv C^k_{(i_A,i_B), j_B} - (0, (A^k_r)_{i_A}) = ((A^{k}_r)_{i_A} + B^{k}_{i_B, j_B}, -(A^{k}_r)_{i_A} - B^k_{i_A, j_B})$. Consider now the two-player, zero-sum infinitely repeated game with one-sided incomplete information and undiscounted payoffs where the stage-game payoff matrices are given by $(D^{1}, D^2)$ and the prior of state $1$ is $q^0 = 1/2$. This modification leaves the payoffs of player 1 unaltered, and therefore does not change his best-reply correspondence when compared to $\mathbb{G}(q^0)$. Though the payoffs of player 3 are modified, his best-reply correspondence is not, which finally implies that the equilibria under this modification are the same as in  $\mathbb{G}(q^0)$. Since the modified game is now zero-sum, it follows from Aumann et al. that it has a (uniform) value, which is the unique ex-ante uniform equilibrium payoff to player 1, and is then evidently the lower end of $I(p^0)$ (cf. footnote 7).

\end{remark}
\end{example}

\section{Information Spillover in Bayesian Persuasion: A comparison with our results}\label{BP}

The problem of information spillover can also be studied in the Bayesian Persuasion ($BP$) setting. This possibility is briefly discussed in the paper by Kamenica and Gentzkow \cite{KG2011} in the section ``Multiple Receivers''. We would like to draw a comparison between the effects of information spillover over equilibrium payoffs in our model and over equilibrium payoffs in $BP$. We first describe the game form of the $BP$ model we have in mind in detail. We refer to this model as \textit{public $BP$}. 

Let $K_A \times K_B$ be the set of \textit{states}, with $K_A$ and $K_B$ finite sets. The set $M_A \times M_B$ is the set of \textit{messages}, with $M_A$ and $M_B$ being finite sets and $|M_i| \geq |K_i|, i \in \{A,B\} $. The actions of player 2 (respec. player 3) are denoted $j_A \in J_A$, ($j_B \in J_B$), with both $J_A$ and $J_B$ being finite sets.  At an ex-ante stage, player 1 chooses a \textit{state-dependent lottery} or \textit{experiment} $x \in \D(M_A \times M_B)^{K_A \times K_B}$ (the set of pure actions of Player 1). Then Nature draws a state $(k_A, k_B)$ according to some prior probability $p^0 \in \D(K_A \times K_B)$ and a message $(m_A, m_B)$ according to $x^{(k_A,k_B)} \in \D(M_A \times  M_B)$. Players 2 and 3 observe the message $(m_A, m_B)$ but not the states; player 2 takes an action $j_A$ and player 3 takes an action $j_B$ and the game ends. We now define payoffs for the players. Given state $(k_A, k_B) \in K_A \times K_B$ and actions $j_A$ and $j_B$ of players 2 and 3, player 1 obtains payoff $u_A(k_A, j_A) + u_B(k_B, j_B)$; player 2 obtains $\nu_A(k_A, j_A)$ and player 3 obtains $\nu_B(k_B, j_B)$. 

In the model just described, players 2 and 3 observe messages \textit{publicly} (which motivates the terminology public $BP$). We will also be interested in the model where messages are observed \textit{privately} by each player (i.e., player 2 observes only $m_A$ and player 3 observes only $m_B$), and will call this model \textit{private $BP$}. 

Using the equilibrium concept in \cite{KG2011}, it is not hard to show that any equilibrium will pay to player 1 the same payoff. We compute this equilibrium payoff: for any $q_A \in \D(K_A)$, let $\tau_A(q_A) = \text{argmax}_{\t_A \in \D(J_A)}\sum_{k_A}\sum_{j_A}q^{k_A}_A \tau_A(j_A)\nu_A(k_A, j_A)$; for $q_B \in \D(K_B)$ we define $\tau_B(q_B) \in \D(J_B)$ analogously for player $3$. Each $x \in \D(M_A \times M_B)^{K_A \times K_B}$ uniquely corresponds to a distribution over $\D(K_A \times K_B)$ with finite support and with mean $p^0$, so we can without loss assume that player 1 chooses a distribution over $\D(K_A \times K_B)$ with finite support and mean $p^0$. Concretely, this amounts to choosing a vector $\l = (\l_m)_{m \in M_A \times M_B}, \l_m \geq 0$ and $\sum_{m}\l_m = 1$ and $(p_m)_{m \in M_A \times M_B}$, $p_m \in \D(K_A \times K_B)$ such that $\sum_{m \in M_A \times M_B} \l_m p_m = p^0$. For each realized message $m \in M_A \times M_B$, the induced posterior is denoted $p_m \in \D(K_A \times K_B)$ and $\l_m$ corresponds to the probability with which $p_m$ realizes. In equilibrium, player 1 chooses a distribution over posteriors so as to maximize $\sum_{m \in M_A \times M_B}\l_m(U_A(p_{mA}) + U(p_{mB}))$, where $U_A(q_A) = \sum_{k_A \in K_A}\sum_{j_A \in J_A}q^{k_A}_A \tau_A(q_A)(j_A) u_A(k_A, j_A)$ and  $U_{B}(q_B) = \sum_{k_B \in K_B}\sum_{j_B \in J_B}q^{k_B}_B \tau_B(q_B)(j_B) u_B(k_B, j_B)$. It is now clear that the maximum value of this program corresponds precisely to the definition of Cav$(U_A + U_B)(p^0)$, which is the equilibrium payoff of player 1.

It is easy to construct examples where $\text{Cav}(U_A)(p^0_A) + \text{Cav}(U_B)(p^0_B) > \text{Cav}(U_B + U_A)(p^0)$.\footnote{Take for instance the following non-zero sum public $BP$ example: $K_A = K_B = \{ 1,2\}$ with $p^0 \in \D(K_A \times K_B)$ defined by $(p^0)^{(1,1)} = 1/2$ and $(p^0)^{(2,2)} = 1/2$. Let $J_A = \{j_A, j'_A \}$ and $J_B = \{j_B, j'_B\}$; let $M_1 = M_2 = \{m, m'\}$. Define $u_A(1,j_A)= 0, u_A(1,j'_A) =2, u_A(2,j_A) = 2, u_A(2,j'_A) =0$. For player 2,  $\nu_A(1,j_A) =-1, \nu_A(1, j'_A) = 1, \nu_A(2, j_A) = 1, \nu_A(2, j'_A) = -1$. For player 3, $\nu_B(1,j_B) = u_B(1,j_B) = -1; \nu_B(2, j_B) =u_B(2,j_B) = 1; u_B(1,j'_B) = \nu_B(1, j'_B) = 1; u_B(2,j'_B) = \nu_B(2, j'_B) = -1$. In this example, we have that $\text{Cav}(U_A)(p^0_A) + \text{Cav}(U_B)(p^0_B) > \text{Cav}(U_A + U_B)(p^0)$.} If in addition we assume we are in the public $BP$ model, we have shown in the previous paragraph $\text{Cav}(U_A)(p^0_A) + \text{Cav}(U_B)(p^0_B)$ cannot be an equilibrium payoff, because of the information spillover phenomenon.

When messages are \textit{privately} sent to players, however, this is simply the standard $BP$ model of a sender simultaneously playing two receivers, which implies that $\text{Cav}(U_A)(p^0_A) + \text{Cav}(U_B)(p^0_B)$ is an equilibrium payoff. The difference between $\text{Cav}(U_A)(p^0_A) + \text{Cav}(U_B)(p^0_B) - \text{Cav}(U_B + U_A)(p^0)$ can be interepreted, therefore, as the loss to the Sender generated by information spillover in the public $BP$ model.

If we specify payoffs to be zero sum, i.e., $\nu_A = - u_A$ and $\nu_B = - u_B$, the public $BP$ model yields that $U_A$ and $U_B$ are concave functions, therefore implying that  $\text{Cav}(U_A)(p^0_A) + \text{Cav}(U_B)(p^0_B) = \text{Cav}(U_A + U_B)(p^0) = U_B(p^0_A) + U_A(p^0_A)$. The next claims settles this result.

\begin{claim}In the public $BP$ zero-sum model, $U_A$ as well as $U_B$ are concave. Therefore, for $p^0 \in \D(K_A \times K_B), \text{Cav}(U_A)(p^0_A) + \text{Cav}(U_B)(p^0_B) = \text{Cav}(U_A + U_B)(p^0) = U_B(p^0_B) + U_A(p^0_A)$. \end{claim}

\begin{proof} We prove that $U_A$ is concave. The proof of concavity of $U_B$ is similar. For each $p_A \in \D(K_A)$ and $j_A \in J_A$, let $f_{j_A}(p_A) =  \sum_{k_A \in K_A} p^{k_A}_A u_A(k_A, j_A)$. Note that $f_{j_A}$ is an affine function of $p_A$. Since player $2$ is a minimizer, $U_A(p_A) = \text{min}_{j_A \in J_A}\{f_{j_A}(p_A)\}_{j_A \in J_A}$. The map $U_A$ is therefore piecewise affine and concave in $p_A$. The remainder of the claim is immediate from the definition of Cav. \end{proof}

Hence, the effect of information spillover in the zero-sum public $BP$ model is inexistent from the perspective of equilibrium payoffs, but is relevant in the non-zero sum public $BP$ model, since there might be loss to player 1 generated by information spillover. As we showed with our main result 1 in this paper, for the model $\mathcal{G}(p^0)$ the difference between $\text{Cav}(v_A)(p^0_A) + \text{Cav}(v_B)(p^0_B) - \text{Cav}(\mathfrak{h})(p^0)$ cannot be interpreted similarly as the loss generated by information spillover, because $\text{Cav}(v_A)(p^0_A) + \text{Cav}(v_B)(p^0_B)$ might be attained in equilibrium. 

\section{Conclusion}\label{SEC: conc}

We studied a three-player generalization of the Aumann et al. model and analysed the effects of information spillover on the equilibrium payoff set of the informed player. Our first two results provided a sufficient condition under which a continuum of equilibrium payoffs exist in the model and which implies, in particular, the existence of equilibria where the informed player circumvents the adverse effects of information spillover. These equilibria involve no signalling on equilibrium path. This sufficient condition is implied by the more interpretable local non-revelation condition. Our second main result presented a necessary condition for equilibria to attain the upper end of $I(p^0)$, which provides a restriction on the signalling processes that can be generated by such an equilibrium. A corollary of this result is that equilibria which ``uncorrelate'' the two two-player zero-sum games $G_A(p^0_A)$ and $G_B(p^0_B)$(whenever $I(p^0)$ is, of course, non-degenerate) do not achieve the upper end of $I(p^0)$.

Several questions remain unanswered with regards to the model $\mathcal{G}(p^0)$. What are the ex-ante equilibrium payoffs of the informed player that can be achieved through signalling on path of play? Is it possible, when $\mathcal{NR}(p^0) = \emptyset$ and $I(p^0)$ is non-degenerate, that the upper end of $I(p^0)$ is achieved as an equilibrium payoff of the informed player? As our last example (Example \ref{not attainable}) showed, it might be the case that only the lower end of $I(p^0)$ is achievable as an equilibrium payoff. Is it possible that an example of $\mathcal{G}(p^0)$ exists for which the upper end of $I(p^0)$ is not achievable, but something in the interior of $I(p^0)$ is an equilibrium payoff? These questions remain to be answered in future work.

\section{Appendix A} 

 \subsection*{Proof of Theorem \ref{main1}}\begin{proof}
For the  proof we maintain the notation for the parametrization $T$, which is established before the statement of Theorem \ref{main1}. Let $\phi_A$ and $p_A$ be the vector and probability distribution over $\D(K_A)$ respectively, given by the property $NR$ at $p^0_A$. We will show that $\phi_A \in NR_A(p^0_A)$. We first show that $\phi_A$ satisfies (i) and (ii) of $NR_A(p^0_A)$. The vector $\phi_A$ satisfies (1) of Definition \ref{main def}, which implies it satisfies (ii) of $NR_A(p^0_A)$ immediately. For (i), let $x \in \Re^{|K_A|-1}, T(x) = p_A$ and $h \in \Re^{|K_A|-1}$ such that $x+h \in P$. Because $\phi_AS$ satisfies (3) of Definition \ref{main def}, we have that $\Cav(v_A \circ T|_{P})(x)+ \phi_A S \cdot h \geq \Cav(v_A \circ T|_{P})(x+h)$. From (1) in Definition \ref{main def} and the definition of $\Cav(v_A)$, $(v_A \circ T|_{P})(x) + \phi_AS \cdot h  \geq (v_A \circ T|_{P})(x+h)$. Again, from (1), the left hand side of the last inequality can be written as $\phi_A \cdot T(x) + \phi_AS \cdot h = \phi_A \cdot S(x+h) + \phi_A \cdot e^{K_A}_{1} = \phi_A \cdot T(x+h)$. Therefore, we have that $\phi_A \cdot T(x+h) \geq (v_A \circ T|_{P})(x+h)$, for any $h \in \Re^{|K_A|-1}$ such that $x+h \in P$. Hence, $\phi_A \cdot q \geq v_A(q), \forall q \in \D(K_A)$. This proves $\phi_A$ satisfies (i). 

It now remains to prove $\phi_A \in F_A$. We will show that $\partial v_A(p_A) \subseteq F_A$.\footnote{We thank an anonymous referee for suggesting a simpler proof of this claim.} Let $f = v_A \circ T$, $x \in P$ and $h \in \Re^{|K_A|-1}$. Let $\S (x)$ denote the set of optimal strategies of the maximizer in the one-shot zero-sum game with matrix $A(T(x))$ and analogously denote $\T (T(x))$ for the optimal strategies of the minimizer in the same game. Applying Proposition 3.4.2 in \cite{LRS2019}:

$$ \lim_{\e \to 0^+} \frac{1}{\e}(f(x + \e h) - f(x)) = \max_{\s \in \S(T(x))}\min_{\t \in \T(T(x))} \s A(S h) \t' = \max_{\s \in \S (T(x))}\min_{\t \in \T(T(x))} (\s A^{k_A} \t')_{k_A \in K_A} \cdot Sh.$$

Let $B$ be the closed unit ball in $\Re^{|K_A|-1}$. If $f$ is differentiable at $x$, the above result implies that 

$$\max_{h \in B} \min_{\phi_A \in F_A} (\nabla f(x) - \phi_A S) \cdot h \leq 0. $$ Applying the minmax theorem gives now that $\nabla f(x) \in F_A$. Let $T(\bar x) = p_A$. By definition of the generalized gradient of $f$ at $\bar x$, we have that $\partial v_A(p_A) \subseteq F_A$, as we wanted to show. \end{proof}

\subsection*{Proof of Proposition \ref{corgeometric}}

\begin{proof}The proof is divided in two-cases: (a) $\Cav(v_A)(p^0_A) > v_A(p^0_A)$ and (b) $\Cav(v_A)(p^0_A) = v_A(p^0_A)$. We start with $(a)$. We first construct the candidate vector of payoffs $\phi_A \in \Re^{|K_A|}$ and a probability distribution $p_A$ and show the pair $(\phi_A, p_A)$ satisfies the definition of $NR$ at $p^0_A$.

Let $\Cav(v_A)(p^0_A) = \sum_{s \in S}\l_s v_A(p_s)$, $\l_s \geq 0$, $\sum_{s \in S}\l_s =1$  and let $p_{s_0}$ be the posterior which is interior to the simplex of states. Let $\H = \{(q,\a) \in \D(K_A) \times \Re | \Cav(v_A)(q) \geq \a \}$. Note that $\H$ is a convex subset of $\Re^{|K_A|} \times \Re$. Now, note that $(p^0_A, \Cav(v_A)(p^0_A)) = \sum_{s \in S}\l_s(p_s, v_A(p_s))$. Since, by (a), $\forall s \in S, p_s \neq p_0$, it follows that $(p^0_A, \Cav(v_A)(p^0_A))$ is not an extremum point of $\H$ (cf. \cite{TR1970}, Section 18). Therefore, there exists a face $F$ of $\H$, with dimension $d \geq 1$, such that $F$ contains $(p^0_A, \Cav(v_A)(p^0_A))$ and each point $(p_s, v_A(p_s))$. Take now a supporting hyperplane $H$ to $\H$ which contains $F$. The hyperplane $H$ intersected with $\D(K_A) \times \Re$ is the graph of an affine function $q \in \D(K_A) \mapsto \phi_A \cdot q \in \Re$, where $\phi_A$ is a vector in $\Re^{|K_A|}$. The candidate pair satisfying $NR$ at $p^0_A$ is $(\phi_A, p_{s_0})$. We now check that it satisfies the conditions of Definition \ref{main def}. It is clear from the construction that $\phi_A \cdot q \geq \Cav(v_A)(q), \forall q \in \D(K_A)$. Using the definition of $T$, we can rewrite this inequality as $\Cav(v_A \circ T|_{P})(x) + \phi_AS \cdot h \geq \Cav(v_A \circ T|_{P})(x+h), T(x) = p_{s_0}, \forall h \in \Re^{|K_A|-1}$ such that $x+h \in P$. We obtain $\phi_A S \in \partial^* \Cav(v_A)(p_{s_0})$. Since

\begin{equation} \tag{*} \phi_A \cdot p_{s_0} = \Cav(v_A)(p_{s_0}) = v_A(p_{s_0}) \end{equation} \begin{equation} \tag{**} \Cav(v_A)(p^0_A) = \phi_A \cdot p^0_A, \end{equation} (1) in Definition \ref{main def} is satisfied. It remains to prove that $\phi_A S \in \partial v_A(p_{s_0})$. In order to show this, we are going to apply a result in Clark \cite{FC1975}.\footnote{We thank an anonymous referee for pointing out a simplification of the proof of this step.} Let $v_A \circ T = f$. Note that $f$ has directional derivatives in all directions (see Proposition 3.4.2 in \cite{LRS2019}). We denote the directional derivative of $f$ at point $x$ in the direction $h$ by $f'(x; h)$. 

Notice now that from (*) and (**) we have that for each $h \in \Re^{|K_A|-1}$, $(-\phi_A) S  \cdot h \leq (-f)'(x; h)$. Corollary 1.10 in Clark \cite{FC1975} now implies that $-\phi_A S \in \partial (-f)$. Since we have that $\partial (f)(x)=-\partial (-f)(x)$, it follows immediately that $\phi_A S \in \partial v_A (p_{s_0})$. This concludes the proof of case (a).

Now we prove case (b). Assume $\Cav(v_A)(p^0_A) = v_A(p^0_A)$. From Aumann et al. \cite{AMS1995}, an optimal strategy $\s_A$ of the informed player in $G_A(p^0_A)$ is to play the mixed action which is optimal in the (one-shot) zero-sum game with matrix $A(p^0_A)$, independently at each stage, whereas the strategy of the uninformed player $2$ is an approachability strategy $\t_A$ at $p^0_A$. By definition of the approachability strategy, there exists a vector $\phi_A \in \Re^{|K_A|}$ satisfying $\phi_A \cdot q \geq v_A(q), \forall q \in \D(K_A)$ such that player $2$ approaches $\phi_A - \Re^{|K_A|}_+$ and the profile $(\s_A, \t_A)$ is a uniform equilibrium, with associated vector of payoffs equal to $\phi_A$. We now observe that the pair $(\phi_A, p^0_A)$ satisfies the conditions of $NR$ at $p^0_A$. First, $\phi_A \cdot p^0_A = \Cav(v_A)(p^0_A) = v_A(p^0_A)$ (which is condition (1)) is clear. By the same reasoning as in the first paragraph of this proof $\phi_A \cdot q \geq v_A(q), \forall q \in \D(K_A)$ is then equivalent to $\phi_AS \in \partial^* \Cav(v_A)(p^0_A)$, which is condition (3). The same reasoning as in the paragraph above (changing now $x = T(p_{s_0})$ to $x^0 = T(p^0_A)$) gives that $\phi_A \in \partial v_A(p^0_A)$, which is condition (2). This concludes the proof. \end{proof}

\section{Appendix B}

\subsection*{Proof of Lemma \ref{characterization}}

Let $(\s, \t_A, \t_B)$ be an equilibrium profile of $\mathcal{G}(p^0)$. We start with the construction of the sequence of random variables $(p_t, \b_{A,t}, \b_{B,t})_{t \in \mathbb{N}}$. The sequence $(p_t)_{t \in \mathbb{N}}$ is the martingale of posteriors obtained from $(\s, \t_A, \t_B)$. It is therefore, immediately adapted to the sequence of increasing fields $(\H_t)_{t \in \mathbb{N}}$. Fix now a Banach limit $L: \ell^{\infty} \to \Re$. \footnote{For the definition of Banach limit, see Section 4.2 in Hart \cite{SH1985}.} A triple of strategies $(\sigma, \tau_{A}, \tau_{B})$ is an $L$-equilibrium if:

\begin{enumerate}
\item $L((\a^{k_A, k_B}_T(\s, \t_A, \t_B))_{T \in \mathbb{N}}) \geq L((\a^{k_A, k_B}_T(\s',\t_A, \t_B))_{T \in \mathbb{N}}), \forall \s \in \S, (k_A, k_B) \in K_A \times K_B$. 
\item $L((\b^A_T(\s, \t_A, \t_B))_{T \in \mathbb{N}}) \geq L((\b^A_T(\s,\t'_A, \t_B))_{T \in \mathbb{N}}), \forall \t'_A \in \T_A$. 
\item  $L((\b^B_T(\s, \t_A, \t_B))_{T \in \mathbb{N}}) \geq L((\b^B_T(\s,\t_A, \t'_B))_{T \in \mathbb{N}}), \forall \t'_B \in \T_B$. 
\end{enumerate}

For notational convenience, we shall denote $L((u_t)_{t \in \mathbb{N}})$ by $L(u_t)$.  A uniform equilibrium of the game $\mathcal{G}(p^0)$ automatically satisfies (1), (2) and (3) above, so it is an $L$-equilibrium.  For each $t \in \mathbb{N}$, let  $\beta_{A,t} = L(\mathbb{E}[\alpha_T|\mathcal{H}_t])$, where $\alpha_T = \frac{1}{T}\sum \limits_{t=1}^T A^{\kappa_A}_{i^t_A, j^t_A}$ and $\beta_{B,t} = L(\mathbb{E}[\beta_T|\mathcal{H}_t])$, where $\beta_T = \frac{1}{T}\sum \limits_{t=1}^T B^{\kappa_B}_{i^t_B,j^t_B}$. We now show (1) of Lemma \ref{characterization}: we have already argued that $(p_t)_{t \in \mathbb{N}}$ is a martingale and it is immediate it is bounded. Fix now $s <t$ and $h_s \in H_s$ with $\mathbb{P}(h_s)>0$: $\mathbb{E}[\b_{A,t}|\mathcal{H}_s](h_s)$ = $\mathbb{E}[L(\mathbb{E}[\a_T| \mathcal{H}_t])| \mathcal{H}_s](h_s) = \frac{1}{\mathbb{P}(h_s)}\mathbb{E}[L(\mathbb{E}[\a_T| \mathcal{H}_t])\mathbbm{1}_{h_s}]$. Letting now $X_T := \mathbb{E}[\a_T| \mathcal{H}_t]\mathbbm{1}_{h_s}$, we have that $X := (X_T)_{T \in \mathbb{N}}: \Omega \to \ell^{\infty}$ is a random variable that takes finitely many values in $\ell^{\infty}$. Therefore, by Lemma 4.6 in Hart \cite{SH1985}, $L$ commutes with $\mathbb{E}$ and we obtain: $\frac{1}{\mathbb{P}(h_s)}\mathbb{E}[L(X_T)] = \frac{1}{\mathbb{P}(h_s)}L(\mathbb{E}[X_T])$ = $\frac{1}{\mathbb{P}(h_s)}L(\mathbb{E}[\mathbb{E}[\alpha_T|\mathcal{H}_t]\mathbbm{1}_{h_s}]) = L(\frac{1}{\mathbb{P}(h_s)}\mathbb{E}[\mathbb{E}[\alpha_T|\mathcal{H}_t]\mathbbm{1}_{h_s}])$, where the last equality follows from $L$ being linear. Now, $L(\frac{1}{\mathbb{P}(h_s)}\mathbb{E}[\mathbb{E}[\alpha_T|\mathcal{H}_t]\mathbbm{1}_{h_s}]) = L(\mathbb{E}[\mathbb{E}[\a_T| \mathcal{H}_t]| \mathcal{H}_s](h_s))$, from the definition of the conditional expectation on a finite field. As $h_s$ was arbitrarily chosen, we have $\mathbb{E}[\b_{A,t}|\mathcal{H}_s] = L(\mathbb{E}[\a_T| \mathcal{H}_s]) = \b_{A,s}$ a.s., proving $(\b_{A,t})_{t \in \mathbb{N}}$ is a martingale adapted to $(\mathcal{H}_t)_{t \in \mathbb{N}}$(boundedness of this martingale is immediate from the fact that payoffs are bounded). The proof that $(\b_{B,t})_{t \in \mathbb{N}}$ is a bounded martingale adapted to $(\mathcal{H}_t)_{t \in \mathbb{N}}$ is the same. We now show condition (2) of Lemma \ref{characterization}. Notice that $\b_{A,1} = L(\mathbb{E}[\a_T| \mathcal{H}_1]) = L(\mathbb{E}[\a_T]) = \lim_{T \to +\infty}\mathbb{E}[\a_T]$, where the last equality follows from the fact that payoffs of an equilibrium profile converge, by definition; the limit $\lim_{T \to +\infty}\mathbb{E}[\a_T]$ is precisely the expected payoff to player 1 in the game $G_A(p^0_A)$ from the equilibrium. By the same reasoning, $\b_{B,1}$ is precisely the expected payoff of player 1 in game $G_B(p^0_B)$. 

We now prove (3) from Lemma \ref{characterization}. Aiming for a contradiction, suppose there exists $h^0_t \in H_t$ with $\mathbb{P}(h^0_t) >0$, such that $\b_{A,t}(h^0_t) > \Cav(v_A)(p_{tA})(h^0_t)$. We will construct a profitable deviation for player 2. Let $\t'_A$ be the following strategy. After $h^0_t$ has occurred, $\t'_A$ will be equal to a Blackwell strategy $\hat \t_A(p_t)$ for the uninformed player in $G_A(p_{t A})$, where $p_t := p_t(h^0_t)$ is the posterior given $h^0_t$; otherwise, $\t'_A$ is equal to $\t_A$. We have, 
$$\mathbb{E}[\a_T] - \mathbb{E}'[\a_T] = \mathbb{P}(h^0_t)(\mathbb{E}[\a_T|h^0_t] - \mathbb{E}'[\a_T|h^0_t]),$$
where $\mathbb{E}'$ is the expectation taken with respect to the $p^0, \s, \t'_A, \t_B$. We claim that $L(\mathbb{E}'[\a_T|h^0_t]) \leq \Cav(v_A)(p_{tA}(h^0_t)), \forall t \in \mathbb{N}$. We will assume the claim for now and conclude the proof of $(3)$. We then prove the claim. Taking $L$ on the right-hand-side of the highlighted equation and using linearity of $L$, one gets: $\b_{A,t}(h^0_t) - L(\mathbb{E}'[\a_T|h^0_t]) \geq \b_{A,t}(h^0_t) - \Cav(v_A)(p_{tA})(h^0_t)>0$. This implies a contradiction: since taking $L$ on the left hand side we obtain the difference between the expected payoffs from $(\s,\t_A,\t_B)$ to player 1 and the expected payoffs from $(\s,\t'_A,\t_B)$ to player 1. As the expected payoffs to player $1$ decreased strictly, player 2 is strictly better, which is a contradiction with the fact that $(\s,\t_A,\t_B)$ is an equilibrium. 

We now prove the claim: note that, $\mathbb{E}'[\a_T |h^0_t] = \frac{\mathbb{E}'[\a_T.\mathbbm{1}_{h^0_t}]}{\mathbb{P'}(h^0_t)} = \frac{1}{\mathbb{P}'(h^0_t)}\frac{1}{T}[\int_{h^0_t}\sum^{t-1}_{s=1}A^{\kappa_A}_{i^s_A, j^s_A} dP'] + \frac{1}{\mathbb{P}'(h^0_t)}\frac{T-t}{T}[\int_{h^0_t} \frac{1}{T-t}\sum^T_{s=t}A^{\kappa_A}_{i^s_A, j^s_A} dP']$. Let $C_T := \frac{(T-t)}{T}$ and $K_T :=  \frac{1}{\mathbb{P}'(h^0_t)}\frac{1}{T}[\int_{h^0_t}\sum^{t-1}_{s=1}A^{\kappa_A}_{i^s_A, j^s_A} dP']$. Then, $\mathbb{E}'[\a_T |h^0_t]  = \frac{1}{\mathbb{P}'(h^0_t)}C_T[\int_{h^0_t} \frac{1}{T-t}\sum^T_{s=t}A^{\kappa_A}_{i^s_A, j^s_A} dP'] + K_T$. Now letting $P^{\circ}(\cdot) := P'(\cdot | h^0_t)$ and $\mathbb{E}^{\circ}[\cdot]$ the expectation operator with respect to $P^{\circ}(\cdot)$, we can re-write $\mathbb{E}'[\a_T |h^0_t]$ as $C_T \mathbb{E}^{\circ}[\frac{1}{T-t}\sum^{T}_{s=t}A^{\kappa_A}_{i^s_A,j^s_A}] + K_T$. Let now $\e>0$. Note that because $\hat \t_A(h^0_t)$ is a Blackwell strategy, there exits $T_0 \in \mathbb{N}$ s.t. $\forall T \geq T_0$, 
 $\mathbb{E}^{\circ}[\frac{1}{T-t}\sum^{T}_{s=t}A^{\kappa_A}_{i^s_A,j^s_A}] \leq \Cav(v_A)(p_{tA}) + \e$. Therefore, $C_T \mathbb{E}^{\circ}[\frac{1}{T-t}\sum^{T}_{s=t}A^{\kappa_A}_{i^s_A,j^s_A}] + K_T \leq C_T \Cav(v_A)(p_{tA}(h^0_t)) + \e C_T + K_T$. Since $K_T \to 0$ and $C_T \to 1$ (as $T \to +\infty$), taking $L$ on both sides of the inequality gives $L(\mathbb{E}'[\a_T| h^0_t ]) \leq \Cav(v_A)(p_{tA}(h^0_t)) + \e$. As $\e>0$ is arbitrary, $L(\mathbb{E}'[\a_T| h^0_t ]) \leq \Cav(v_A)(p_{tA}(h^0_t))$, as required. This concludes the proof of the claim and the proof of (3). The proof of (4) is the exact same. 

\begin{proof}[Proof of the Claim \ref{aux claim}]
Fix $t,s \in \mathbb{N}$ with $s \geq t$. Conditioning over $\mathcal{H}_{s+1}$ and $\kappa_A$, we have $$\mathbb{E}[A^{\kappa_A}_{i^s_A,j^s_A} | \mathcal{H}_t] = \mathbb{E}[ \sum_{k_A \in K_A}p^{k_A}_{s+1 A}A^{k_A}_{i^s_A,j^s_A} | \mathcal{H}_t] = \sum_{k_A \in k_A}p^{k_A}_{tA}\mathbb{E}[A^{k_A}_{i^s_A, j^s_A} | \mathcal{H}_t] + \sum_{k \in K}\mathbb{E}[(p^{k_A}_{s+1A} - p^{k_A}_{tA}) A^{K_A}_{i^s_A, j^s_A} | \mathcal{H}_t] =$$  $$=\mathbb{E}[A(p_{tA})_{i^s_A,j^s_A} | \mathcal{H}_t] + \sum_{k_A \in K_A}\mathbb{E}[(p^{k_A}_{s+1A} - p^{k_A}_{tA}) A^{k_A}_{i^s_A, j^s_A} | \mathcal{H}_t] \leq V_A(p_{tA}) + \sum_{k_A \in K_A}\mathbb{E}[(p^{k_A}_{s+1A} - p^{k_A}_{tA}) A^{k_A}_{i^s_A, j^s_A} | \mathcal{H}_t].$$ Summing over $s \in \mathbb{N}$ from  $t \leq s \leq T$, and dividing by $T$, we have $$\mathbb{E}[\alpha_T | \mathcal{H}_t] \leq \frac{t}{T} + V_A(p_{tA}) + \frac{1}{T}\sum_{t \leq s \leq T}\sum_{k_A \in K_A} \mathbb{E}[|p^{k_A}_{s+1A} - p^{k_A}_{tA}||\mathcal{H}_t],$$ where we used the fact that stage payoffs at any state are bounded $1$.  Denote $Z_t := \frac{t}{T} + \frac{1}{T}\sum_{t \leq s \leq T}\mathbb{E}[ \pi_t |\mathcal{H}_t]$, where $\pi_t := \sum_{k_A \in K_A}\sup_{s \geq t}|p^{k_A}_{s+1A} - p^{k_A}_{tA}|$. Taking Banach-limits (on $T$) on both sides, we have $\beta_{A, t} \leq V_A(p_{tA}) + \mathbb{E}[\pi_t | \mathcal{H}_t]$ a.s..  Since $p_t \rightarrow p^{\infty}$ a.s., as $t \rightarrow \infty$,  it follows by Lemma 4.24 in \cite{SH1985}, that $\mathbb{E}[\pi_t | \mathcal{H}_t] \rightarrow 0$ a.s..\end{proof}

\section{Supplemental Appendix}

Let $K$, $I$ and $J$ be finite sets, with $K$, $I$ and $J$ with cardinality larger than or equal to 2.\footnote{If $I$ has cardinality 1, than the nonrevealing value function is concave. Therefore it is equal to its concavification. If $J$ has cardinality 1, then one optimal strategy of the informed player is to completely reveal the information at any prior, since the uninformed player has no other strategy to play. This implies the concavification is affine.}  A collection of matrices $(A^{k})_{k \in K}$, with $A^{k} \in \mathbb{R}^{I \times J}$ defines a unique zero-sum game with lack of information on one side $G_A(p)$, for some $p \in \Delta(K)$. Let $\mathcal{G}$ be the set of vectors $((A^{k})_{k \in K},p) \in \prod_{k \in K}\mathbb{R}^{(I \times J)} \times \Delta(K)$. Let $C(\Delta(K))$ denote the class of continuous real maps in $\Delta(K)$ and endow $C(\Delta(K))$ with the $|| \cdot||_{\infty}$-norm. The next three results are known (see \cite{SS2002}) and are stated for future reference.  

\begin{lemma}\label{optimal nonrevealing}
Let $(A^1,...,A^{|K|},p)$ define a zero-sum game $G_A(p)$ with lack of information on one side. There exists an optimal strategy of the informed player such that the induced posteriors $(p_s)_{s \in S}$ by such strategy satisfy $|S| \leq |K|+1$. 
\end{lemma}

\begin{lemma}\label{v is continuous}
Let $f: \prod_{k \in K}\mathbb{R}^{(I\times J)} \rightarrow C(\Delta(K))$ be defined by $f(A^1, A^2,...,A^{|K|}) =  v_A$. Then $f$ is continuous.
\end{lemma}

\begin{theorem}\label{cav is continuous}
The operator $\text{Cav}: C(\Delta(K)) \rightarrow C(\Delta(K)) $ \footnote{Laraki \cite{RL2004} shows the operator \text{Cav} is well defined.} is continuous. 
\end{theorem}

\subsection{Nondegeneracy of $I(p^0)$ is Robust to payoff perturbations}

The next result shows that the property of the interval $I(p)$ being nondegenerate is robust to stage-payoffs perturbations. 

\begin{proposition}\label{nondeg}Let $I_A, J_A,I_B,J_B, K_A$ and $K_B$ be finite sets. Let  $\mathcal{G}$ denote the set of vectors $$(A^1, A^2,...,A^{|K_A|}, B^1, B^2,...,B^{|K_B|}) \in \prod_{k_A \in K_A}\mathbb{R}^{(I_A \times I_B)} \times \prod_{k_B \in K_B}\mathbb{R}^{(I_B \times J_B)}$$ that define a  $\mathcal{G}(p)$ for which $I(p)$ is nondegenerate. Then  $\mathcal{G}$ is open in $\prod_{k_A \in K_A}\mathbb{R}^{(I_A \times I_B)} \times \prod_{k_B \in K_B}\mathbb{R}^{(I_B \times J_B)}$. 
\end{proposition}

\begin{proof}  Recall that $\mathfrak{h}(q): \D(K_A \times K_B) \to \Re$ was defined as $\mathfrak{h}(q) = v_A(q_A) + v_B(q_B)$. Define $(f + g):  \prod_{k_A \in K_A}\mathbb{R}^{(I_A \times I_B)} \times  \prod_{k_B \in K_B}\mathbb{R}^{(I_B \times J_B)} \rightarrow C(\Delta(K_A \times K_B))$ by $$(f+g)(A^1, A^2,...,A^{|K_A|}, B^1, B^2,...,B^{|K_B|}) = h.$$ Notice that $(f + g)$ is continuous, by Lemma \ref{v is continuous}.

Let now $(A^1_n,...,A^{|K_A|}_n, B^1_n,...,B^{|K_B|}_n)_{n \in \mathbb{N}} \subset \prod_{k_A \in K_A}\mathbb{R}^{(I_A \times I_B)} \times  \prod_{k_B \in K_B}\mathbb{R}^{(I_B \times J_B)}$ be a sequence converging to a point $$(A^1, A^2,...,A^{|K_A|}, B^1, B^2,...,B^{|K_B|})$$ in $\mathcal{G}$. Let $(h_n)_{n \in \mathbb{N}}, (v_{A_n})_{n \in \mathbb{N}}$ and $(v_{B_n})_{n \in \mathbb{N}}$ be the corresponding sequences of nonrevealing value functions. By Lemma \ref{v is continuous}, $h_n$ converges uniformly to $h$, $v_{A_n}$ converges uniformly to $v_A$ and $v_{B_n}$ converges uniformly to $v_B$. It follows now by Theorem \ref{cav is continuous} that there exists $n_0 \in \mathbb{N}$ such that $\forall n \geq n_0$ we have $\text{Cav}(h_n)(p) < \text{Cav}(v_{A_n})(p_A) + \text{Cav}(v_{B_n})(p_B)$. Since the sequence considered is arbitrary, it follows that there exists an open set around $(A^1, A^2,...,A^{|K_A|}, B^1, B^2,...,B^{|K_B|})$ where $I(p)$ is nondegenerate. \end{proof}

\subsection{The property of Local Nonrevelation at the Prior is not Non-Generic}\label{propnongen}

Define $\mathcal{G}_p$ as the set of matrices $(A^1,...,A^{|K|})$ such that $(A^1,...,A^{|K|}, p)$  is locally nonrevealing at $p$. We will show that $\mathcal{G}_p$ has non-empty interior in $\prod_{k \in K}\mathbb{R}^{I \times J}$. We will consider wihtout loss of generality $p \in \text{int}(\Delta(K))$. If the prior $p$ is not in the interior of $\Delta(K)$, then we drop the type that has probability zero from set $K$.

Let $\{f_i\}^{n}_{i=1}$ be a finite collection of real affine functions defined on $\Delta(K)$ such that $f_i(p) := a_i \cdot p + b_i$. This collection defines a concave piecewise linear function by letting $H(p) :=  \min\{f_1(p),...,f_n(p) \}$. This concave piecewise linear function $H$ induces a \textit{polyhedral subdivision}\footnote{See \cite{LSR2010} for a definition of polyhedral subdivision.} on $\Delta(K)$ by projecting the faces of the graph of $H$ over $\Delta(K)$. Consider a point $p_0$ in the interior of $\Delta(K)$ and consider the following polytope of $\mathbb{R}^K \times \mathbb{R}$: let $y_0 = (p_0, x) \in \Delta(K) \times \mathbb{R}$ where $x >0$ and let $v^k = (v_k,0)$ where $v_k$ is a vertex of $\Delta(K)$. Consider $P = \text{co}\{y_0, v^1,...,v^k \}$. The boundary of $P$ minus int($\D(K)) \times \{0\}$ is the graph of a concave piecewise linear function $H_{y_0}$, with each maximal proper face of $P$ corresponding to the graph of an affine function. The polyhedral subdivision on $\Delta(K)$ induced by this concave piecewise linear function will be denoted $\mathcal{P}_{y_0}$.

\begin{proposition}\label{existence of locally nonrevealing at p game}
 Assume $|I|, |J| \geq |K|$ and $p \in \text{int}(\D(K))$. Then there exists a game $(A^1,...,A^{|K|}) \in \mathcal{G}_p$. \end{proposition}

\begin{proof}We prove a slightly stronger result: we will show that there exists $(A^1,...,A^{K},p)$ such that for \textit{any} $(p_i)^n_{i=1} \subset \Delta(K)$ with $\sum^n_{i}\alpha_ip_i = p$ and $\text{Cav}(v_A)(p) = \sum^{n}_{i=1}\alpha_i v_A(p_i)$, there exists $i_0$ such that $p_{i_0} \in \text{int}(\Delta(K))$.

Choose $p_0 \neq p$ in the interior of $\Delta(K)$ and consider the polytope $P$ constructed just like in the previous paragraph. As discussed in the previous paragraph, consider the finite collection of affine functions $\{f_i\}^{J}_{i=1}$ defining the concave piecewise linear function $H_{y_0}$ whose graph is the boundary of polytope $P$. We show that the function $H_{y_0}$ is the nonrevealing value of a game in $\mathcal{G}_p$. Let $a^{t}_j$ be t-th entry of vector $a_j$ such that $f_j(p) = a_j \cdot p + b_j$. Define the matrix $A^t = [v^t_1,...,v^t_J]$  and $\bar \eta = \min_{j,t} \{a^t_j\}$, where $v^t_j$ is a column vector of length $I$ with entries from 1 to $|K|$  all equal to $a^{t}_j + b_j$ and $\eta <  \bar \eta$ everywhere else. Consider the vector $(A^1,...,A^{|K|},p)$. Notice that for each $t$, the rows of the submatrix of $A^t$ composed by the first $|K|$ rows and $|J|$ columns are all equal. So in the one shot zero-sum game given by matrix $A(p)$, the row player is indifferent between these rows, which are in turn strictly better than the rows from $|K|+1$ to $I$, by construction. The column player, who is a minimizer, chooses therefore the columns that minimizes the row players' payoffs: each column $j$ of $A(p)$ is a vector such that the first $|K|$ entries are equal to $f_j(p)$ by construction. So the column player chooses $j_0$ such that $f_{j_0}(p) = \min_{j \in J}\{f_1(p),...,f_J(p)\} = H_{y_0}(p) = v_A(p)$. Therefore the nonrevealing value function $v_A$ equals $H_{y_0}$. Now, let $\mathcal{P}_{y_0}$ be the polyhedral subdivision induced by $H_{y_0}$. By definition $p$ belongs to a maximal cell $C$ of $\mathcal{P}_{y_0}$. If $p$ belongs to the interior of $C$, then any optimal strategy of the informed player at $p$ induces a posterior at the relative interior of $\Delta(K)$, because there is one vertex of $C$ in the interior of $\Delta(K)$, namely $p_0$. If $p$ is in the relative interior of the intersection of two or more cells, then the intersection also has a vertex at $p_0$, which implies that any optimal strategy of the informed player induces at least one posterior in the interior of the simplex. \end{proof}

Proposition \ref{existence of locally nonrevealing at p game} above provides conditions under which $\mathcal{G}_p$ is nonempty. We use it to prove the following robustness result: 

\begin{proposition}
Let $p \in \text{int}(\Delta(K))$ and $|I|, |J| \geq K$. Then $\mathcal{G}_p$ has nonempty interior. 
\end{proposition}

\begin{proof}
Consider $(A^1,...,A^{|K|}, p)$ such that for any $(p_i)^n_{i=1} \subset \Delta(K)$ with $\sum^n_{i}\alpha_ip_i = p$ and $\text{Cav}(v_A)(p) = \sum^{n}_{i=1}\alpha_i v_A(p_i)$, there exists $i_0$ such that $p_{i_0} \in \text{int}(\Delta(K))$. The existence of such a game is guaranteed by the proof of Proposition \ref{existence of locally nonrevealing at p game}. Assume by contradiction there exists a sequence  $(A^1_s,...,A^{|K|}_s)_{s \in \mathbb{N}}$ with $\lim\limits_{s \rightarrow \infty}(A^1_s,...,A^{|K|}_s) = (A^1,...,A^{|K|})$, such that for each $s$, any optimal strategy at $p$ induces posteriors in the boundary of the simplex $\Delta(K)$. By Lemma \ref{optimal nonrevealing}, let $(p^s_1,...,p^s_{|K|+1})$, with $$\sum_{i=1}^{|K|+1} \alpha^s_i p^s_i = p, \hspace{1cm} \sum_{i=1}^{|K|+1} \alpha^s_i =1, \alpha^s_i \geq 0$$ be the vector of posteriors induced by an optimal strategy given by Lemma \ref{optimal nonrevealing} for the zero-sum game with lack of information on one side with prior $p$ defined by $(A^1_s,...,A^{|K|}_s)$. By assumption, $p^s_i \in \partial\Delta(K)$, for all $i=1,..,|K| +1$. Passing to convergent subsequences if necessary, we can assume that $p^s_i \rightarrow p_i \in \partial \Delta (K)$, for each $i =1,...,|K| +1$. Similarly, assume $\alpha^s_i \rightarrow \alpha_i$, for each $i =1,...,|K|+1$. Note that we have $\sum_{i=1}^{|K|+1}\alpha^s_i v_{A_s}(p^s_i) = \text{Cav}(v_{A_s})(p)$. By Theorem \ref{cav is continuous}, $\text{Cav}(v_{A_s})(p) \rightarrow \text{Cav}(v_A)(p)$, as $s \rightarrow \infty$.  By Lemma \ref{v is continuous} and the above assumptions, $\sum_{i}^{|K|+1}\alpha^s_i v_A(p^s_i) \rightarrow \sum_{i}^{|K| +1}\alpha_i v_A(p_i)$, as $s \rightarrow \infty$; it implies that $\sum_{i}^{|K|+1}\alpha_i v_A(p_i) = \text{Cav}(v_A)(p)$. This is a contradiction, since $p_i \in \partial \Delta(K), \forall i = 1,.., |K|+1$. \end{proof}

\subsection{Existence of Equilibria in $\mathcal{G}(p^0)$}

In this section we construct a particular kind of uniform equilibrium of $\mathcal{G}(p^0)$ called a \textit{joint-plan} (see Aumann et al. \cite{AMS1995}) which pays the lower end of $I(p^0)$.

\begin{definition} Let $h^1_m := (i^t_A, i^t_B)_{1 \leq t \leq m-1}$. We define $H^1_n := \bigcup\{h^1_n\}$ and call it the set of individual histories of player 1.
\end{definition}
\begin{definition}\label{definition of joint plan}

An \textit{independent joint-plan}\footnote{The joint-plan is called ``independent'' because each contract is defined by a product of strategies of each player.} in $\mathcal{G}(p)$ is a triple $(S,x,\gamma)$ where: 
\begin{itemize}

\item (Signals) $S$ is the set of signals, i.e., a subset of $H^{1}_n$, for some $n \in \mathbb{N}$. 

\item (Signaling Strategy) The vector $x$ is a $|\supp(p)|$-tuple where for each $(k_A, k_B)$ in $\supp(p)$, $x^{k_A,k_B}$ is a probability distribution on $S$.  

\item (Contracts) $\gamma = (\gamma_A, \gamma_B)$, with $ \gamma_{i} = (\gamma_{i}^{s})_{s \in S}$,   and  $\gamma_i^s := \sigma^s_i \bigotimes \tau^s_i, \sigma^s_i \in \triangle(I_i)$ and $\tau^s_i \in \triangle(J_i)$, for $i \in \{A, B\}$. We denote by $\gamma^s_A(i_A,j_A)$ the probability of moves $(i_A,j_A)$ and $\gamma^s_B(i_B,j_B)$ the probability of moves $(i_B,j_B)$.

\end{itemize} 
\end{definition}

 Following Lemma 2 in \cite{SS1983}, the independent distribution $\gamma^s_A = \sigma^s_A \bigotimes \tau^s_A$ over $I_A \times J_A$ (resp. $\gamma^s_B$) can be induced through the play of a deterministic sequence of moves at each stage by players 1 and 2 (resp. players 1 and 3) with the appropriate frequency. This deterministic path of play is what the contract $\gamma^s_A$ (resp. $\gamma^s_B$) represents.

\subsection*{Notation} Let $(S, x, \gamma)$ be an independent joint-plan in $\mathcal{G}(p)$. We define some notation necessary for the statement of Lemma \ref{joint implies eq}. The prior $p$ and signaling strategy $x$ define a probability distribution $P$ on $K_A \times K_B \times S$ by letting $P(k_A, k_B, s) := p^{k_A, k_B}x^{k_A, k_B}(s)$. We define the \textit{posterior probability} of $(k_A, k_B) \in \supp(p)$ given the realization of $s \in S$ by $p^{k_A,k_B}(s) : = \frac{P(k_A, k_B, s)}{P(s)}$, where $P(s) = \sum_{(k_A,k_B) \in K_A \times K_B}P(k_A,k_B, s)$. This is the probability over states, obtained by Bayes rule, that players 2 and 3 may compute after observing signals. For $s \in S, \ell \in \{A,B\}$, given a posterior probability $p(s) \in \Delta(K_A \times K_B)$, the marginal posterior over $K_\ell$ is $p(s)_{\ell} \in \Delta(K_\ell)$. Also, for each $s \in S$, $(k_A,k_B) \in \supp(p)$, and $\ell, j \in \{A,B\}$, we define: the expected payoff of player 1 after signal $s$ is $\alpha^{k_A,k_B}(s) := \sum_{i_A,j_A,i_B,j_B}(A^{k_A}_{i_A,j_A}\gamma^s_A(i_A,j_A) + B^{k_B}_{i_B,j_B}\gamma^s_B(i_B,j_B))$; the highest payoff player 1 can obtain after any signal when states are $(k_A,k_B) \in \supp(p)$ is $\alpha^{k_A,k_B} := \max_{t \in S}\alpha^{k_A,k_B}(t)$ and $\alpha = (\alpha^{k_A,k_B})_{(k_A, k_B) \in \supp(p)}$ is their vector; the marginal expected payoff after $s$ of player 1 in game $\ell \in \{A,B\}$ is $\alpha^{k_\ell}_{\ell}(s) = \sum_{i_\ell,j_\ell}\ell^{k_\ell}_{i_\ell,j_\ell}\gamma^s(i_\ell,j_\ell)$; we define by $\beta_{\ell}(s) := \sum_{k_{\ell}}p^{k_{\ell}}(s)_{\ell} \alpha^{k_{\ell}}(s)$ the negative of the expected payoff of the uninformed player playing game $\ell$ after signal $s$, and by $\beta_{\ell} := \sum_{s \in S}P(s)\beta_{\ell}(s)$.  

\medskip

\begin{definition}An independent joint-plan $(S,x,\gamma)$ in $\mathcal{G}(p)$ is called \textit{safe} if for each $s \in S$,  $\tau^s_A$ is optimal in the one-shot game with matrix $A(p(s)_A) := \sum_{k_A \in K_A}p^{k_A}(s)A^{k_A}$ and $\tau^s_B$ is optimal in the one-shot zero-sum game $B(p(s)_B)$.\end{definition}


\begin{lemma}\label{joint implies eq}Let $(S,x,\gamma)$ be an independent joint-plan in $\mathcal{G}(p)$ satisfying:

\begin{enumerate}
\item $\beta_A(s) \leq \text{Cav}(v_A)(p(s)_A)$ and $\beta_B(s) \leq \text{Cav}(v_B)(p(s)_B)$, for all $s \in S$.

\item For all $(k_A,k_B) \in \supp(p)$, $s \in S$ such that $P(k_A, k_B, s) > 0$, it implies $\alpha^{k_A,k_B}(s) = \alpha^{k_A,k_B}$.

\item $\alpha \cdot q \geq \mathfrak{h}(q)$, for all $q \in \Delta(\text{supp}(p))$. 
\end{enumerate}
Then there exists an equilibrium $(\sigma, \tau_A, \tau_B) \in \Sigma \times \mathcal{T}_2 \times \mathcal {T}_3$ in $\mathcal{G}(p)$ such that $\forall (k_A,k_B) \in \supp(p)$ we have $\alpha^{k_A, k_B}(\sigma, \tau_A, \tau_B) = \alpha^{k_A, k_B}$, $\beta^A(\sigma, \tau_A, \tau_B) = - \beta_A$ and $\beta^B(\sigma, \tau_A, \tau_B) = - \beta_B$.
\end{lemma}

\begin{proof}Let $(S,x,\gamma)$ be a safe independent joint-plan. We construct an equilibrium $(\sigma, \tau_A, \tau_B)$ in $\mathcal{G}(p)$ inducing the required joint-plan vector of payoffs. The proof is exactly analogous to Proposition 1 in \cite{SS1983}. We indicate the steps of the construction. If $(k_A, k_B) \in \supp(p)$ realizes,  player 1 uses a \textit{state-dependent lottery} (see Lemma 1 in \cite{SS1983}), using finitely many stages of his play to induce one of the posteriors $(p(s))_{s \in S}$. For each $s \in S$, Lemma 2 of \cite{SS1983} implies there exists $h^s_{A, \infty} = ((i^n_A, j^n_A))_{n \geq 1}$ such that for each $(i_A,j_A) \in I_A \times I_B$ $$ \frac{1}{n} | \{m| 1 \leq m \leq n, (i^m_A, j^m_B) = (i_A,j_A)\}| \rightarrow \gamma^s_A(i_A, j_A),$$ as $n \rightarrow + \infty$.

 Assume the state-dependent lottery of player 1 draws signal $s \in S$. Then player 1 will play according to $h^s_{A, \infty}$, as long as player 2 plays according to $h^s_{A, \infty}$.  The asymptotic frequency induced by this deterministic play is $\gamma^s_A$.  The expected payoff to player 2 given that $s$ realizes is then $- \beta_A(s)$. Similarly, players 1 and 3 will play a deterministic sequence of moves according to some $h^s_{B, \infty}$ such that the induced asymptotic frequency is $\gamma^s_B$. The expected payoff to player 3 given that $s$ realizes is then $- \beta_B(s)$. The payoff to player 1 from following the deterministic sequence of moves in both games is $\alpha^{k_A, k_B}(s) = \alpha^{k_A,k_B}$, $\forall s \in S$ that realize with positive probability, by condition (2). If any player deviates from the deterministic sequence of moves they are supposed to play after $s$ realizes, the deviation is detectable. If player 2 deviates, player 1 will punish him by playing the optimal strategy of the zero-sum game $G_A(p_A(s))$. If player 3 deviates, player 1 will punish him by playing the optimal strategy of the zero-sum game $G_B(p_B(s))$. If player 1 deviates from his deterministic sequence in either game he is playing, the deviation is detectable by players 2 and 3 and players 2 and 3 will play the approachability strategy given by Theorem \ref{approach strategy}. During signaling stages, players 2 and 3 can play anything, but if player 1 makes a detectable deviation during signaling stages -- by not using a signal $s$  -- then players 2 and 3 will also play the approachability strategy given by Theorem \ref{approach strategy}. We show that the strategies defined form a uniform equilibrium: first, condition (2) of Theorem \ref{joint implies eq} prevents any undetectable deviation of player 1 at signaling stages from being profitable. After signal $s$ realizes, if player 2 deviates, then player 1 plays the optimal strategy of $G_A(p_A(s))$. Since $\beta_A(s) \leq \text{Cav}(v_A)(p(s)_A)$, it implies that the deviation is not profitable for player 2. The same reasoning applies for a deviation of player 3. The inequality $\alpha \cdot q \geq \mathfrak{h}(q), \forall q \in \Delta(\supp(p))$ of condition (3), shows that a deviation is not profitable for player 1, because $M(\alpha)$ is approachable by players 2 and 3, by Theorem \ref{approach strategy}.

Condition (1) of the uniform equilibrium definition is immediately satisfied, because the payoffs to each player given by the strategies defined above converge. Also, condition (2) of the definition follows immediately from the fact that deviations are punished with approachability strategies, in case player 1 deviates, and optimal strategies of the zero-sum games, in case players 2 or 3 deviate. \end{proof}

Given an independent joint-plan $(S,x, \gamma)$, the vector $(\alpha, -\beta_A, -\beta_B) \in \mathbb{R}^{\supp(p)} \times \mathbb{R} \times \mathbb{R}$ will be called the \textit{vector of payoffs} of the equilibrium joint-plan.

Condition (1) in Lemma \ref{joint implies eq} guarantees that players 2 and 3 do not deviate after signaling stages from the deterministic path of play induced by the joint-plan contracts: in case signal $s$ realizes, the contract $\gamma^s_A$ has expected payoff to player 2 given $s$ of $-\beta_A(s)$ and player 1 can punish player 2 in case player 2 deviates from the deterministic path of play given by the contract $\gamma^s_A$ by playing his optimal strategy at $G_A(p(s)_A)$, which guarantees the payoff of player 2 would not be larger than $-\text{Cav}(v_A)(p(s)_A) \leq - \beta_A(s)$. The analogous reasoning holds to prevent a deviation of player 3. Condition (2) prevents undetectable deviations from player 1 at signaling stages: it says that player 1 cannot profit from ``lying'' about a signal because he is indifferent to the payoffs under any contract that realizes with positive probability, for any pair of states chosen by Nature. Condition (3) implies the existence of strategies for the uninformed players to punish the informed player in case he makes an observable deviation. This is the approachability strategy of the uninformed player in the repeated game $G_{A+B}(p)$.\footnote{Recall $G_{A+B}(p)$ is a two-player zero-sum infinitely repeated game with one-sided incomplete information and undiscounted payoffs.} In section \ref{AS}, we show this strategy can indeed be played by players 2 and 3. This will be a simple consequence of the fact that the payoffs for the informed player have a ``separable'' structure - they are the addition of payoffs obtained in each zero-sum game separately.


\begin{Proposition}\label{existence lower bound}There exists an independent and safe joint-plan in $\mathcal{G}(p^0)$ satisfying (1), (2) and (3) of Lemma \ref{joint implies eq}. Also, if $(\sigma, \tau_A, \tau_B)$ is the equilibrium induced by this joint-plan, then $(\sigma, \tau_A, \tau_B)$ pays $\text{Cav}(h)(p)$ as an ex-ante payoff to the informed player.\footnote{The generalization of Lemma \ref{joint implies eq} to a model of one informed player and $n$ uninformed players -- as can be readily checked in the proof in the Appendix -- is straightforward.}
\end{Proposition}

\begin{proof}
We first state Theorem \ref{SST} and Lemma \ref{lemmaSST} below, which are the main tools for the proof of Proposition \ref{existence lower bound}. 
\begin{theorem}\label{SST}[Simon et al. \cite{SST1995}]
Let $K$ be a finite set and $p^0 \in$ int $\triangle(K)$. Let $a:\triangle(K) \rightarrow \mathbb{R}$ and $h:\triangle(I) \times \triangle(K) \rightarrow \mathbb{R}^{|K|}$ be continuous functions such that:
\begin{enumerate}

\item The function h is affine with respect to the variable $\sigma \in \triangle(I)$, for all $p \in \triangle(K)$.

\item For all $p,q \in \triangle(K)$, there is $\sigma \in \triangle(I)$ such that $h(\sigma,p)\cdot  q \geq a(q)$.

Therefore, there exists a set $P_0 \subset \triangle(K)$ of cardinality $\leq |K|$ and vectors $\sigma_p \in \triangle(I)$ (with $p \in P_0)$ and $\phi \in \mathbb{R}^{|K|}$ such that:

\item $\phi \cdot q \geq a(q)$ for all $q \in \triangle(K)$.

\item $p^0 \in \text{co}P_0$.

\item For all $p \in P_0$, $k \in K$ we have $\phi^{k} \geq h^{k}(\sigma_{p},p)$, with equality occuring in place of inequality whenever $p^k > 0$.

\end{enumerate}

\end{theorem}

We will also make use of the following simple version of a Lemma in \cite{SST1995}:

\begin{lemma}\label{lemmaSST}
For every $\epsilon >0$ there exists a continuous map $g_{A}:\triangle(K_A)\rightarrow \triangle(J_A)$ such that $\sigma_{A} A(p)(g_{A}(p))^{T}\leq v_A(p)+ \epsilon$, for all $(\sigma,p) \in \triangle(I_A) \times \triangle(K_A)$. 
\end{lemma} 

Assume first $p^0 \in \text{int}\Delta(K_A \times K_B)$. Given $\epsilon > 0$, applying Lemma \ref{lemmaSST} we have that $\sigma_{A}A(p_A)(g_{A}(p_A))^{T}\leq v_A(p_A)+ \epsilon$ and $\sigma_{B}B(p_B)(g_{B}(p_B))^{T}\leq v_B(p_B)+ \epsilon$, for all $(\sigma_{A},p_A) \in \triangle(I_A) \times \triangle(K_A)$ and $(\sigma_{B},p_B) \in \triangle(I_B) \times \triangle(K_B)$. Define $h(\sigma,p) = ((m_{A}\sigma)A^{k_A}(g_A(p_{A}))^{T}+(m_{B}\sigma)B^{k_B}(g_B(p_{B}))^{T})_{(k_A ,k_B) \in K_A \times K_B}$, where $m_{A}\sigma := \text{marg}_{I_A}\sigma$ and $m_{B}\sigma := \text{marg}_{I_B}\sigma$. Since the marginal operator is affine, the function $h$ is affine on $\sigma$. It is also continuous. Now, given $p,q \in \triangle(K_A \times K_B)$, let $\bar{\sigma}_A^q$ be the optimal  strategy of the informed player in the one-shot zero-sum game with matrix $A(q_A)$ and let  $\bar{\sigma}_B^q$ be the  optimal strategy of the informed player in the one-shot zero-sum game with matrix $B(q_B)$. Define $\tilde{\sigma} := \bar{\sigma}_A^q \bigotimes \bar{\sigma}_B^q \in \Delta(I_A \times I_B)$. Then we have that $h(\tilde{\sigma},p) \cdot q \geq a(q) := \mathfrak{h} (q) = v_A(q_A) + v_B(q_B)$. Applying Theorem \ref{SST}, we have that there exists $P_0 \subset \triangle(K_A \times K_B)$ with cardinality  $\leq |K_A \times K_B|$ , $(\sigma_{p})_{p \in P_0}$ and $\phi \in \mathbb{R}^{|K_A \times K_B|}$ satisfying (3), (4) and (5). From (3) and (4) we have that there exists a nonnegative collection $(\lambda_{p})_{p \in P_0}$ and a vector $\phi$ such that $\sum_{p\in P_0}\lambda_{p}p=p^0$ and $\sum_{p \in P_0}\lambda_{p}=1$, $\phi \cdot q \geq a(q)$ and $(m_{A}\sigma_{p})A(p_A)(g_A(p_{A}))^{T} \leq v_A(p_A) + \epsilon$ and $(m_{B}\sigma_{p})B(p_B)(g_B(p_{B}))^{T} \leq v_B(p_B) + \epsilon$ for $\sigma_p \in \triangle(I_A \times I_B)$ and $p \in P_0$. 

Notice that the solutions given by the application of Theorem \ref{SST} are all indexed by $\epsilon >0$. For each $n \in \mathbb{N}$, we can therefore consider $P^n_0 = \{p^n_{s}\}_{s \in S_n} \subset \Delta(K_A \times K_B)$ such that $|P^n_0| \leq |K_A \times K_B|$, $(\sigma_{p^n})_{p^n \in P^n_0}$,  $(m_{B}\sigma_{p^n_{s}})$, $g_B(p^n_{sB})$, $(m_{A}\sigma_{p^n_{s}})$, $g_A(p^n_{sA})$, $(\lambda_{p^n_{s}})_{s \in S_n}$ and $\phi_n$ satisfy $(m_{B}\sigma_{p^n_{s}})B(p^n_{sB})(g_B(p^n_{sB}))^{T} \leq v_B(p^n_{sB}) + 1/n$ and $(m_{A}\sigma_{p^n_{s}})A(p^n_{sA})(g_A(p^n_{sA}))^{T} \leq v_A(p^n_{sA}) + 1/n$ such that  $\sum_{s}\lambda_{p^n_{s}} p^n_{s} = p^0$ and $\sum_{s}\lambda_{p^n_{s}}=1$; also, $\phi_n$ such that $\phi_n \cdot q \geq a(q), \forall q$ with (5) being satisfied. 


Passing to a subsequence if necessary, consider the (Hausdorff) limit $P_0$ of the sequence $P^n_0$. Notice that $P_0$ has finite cardinality (less than $|K_A \times K_B|$). We can also consider limits of the associated solutions, since they all lie in compact sets.\footnote{Property (5) of Theorem \ref{SST} guarantees then that the sequence of vectors $(\phi_n)_{n \in \mathbb{N}}$ is bounded, so it will also have an accumulation point.} Therefore, consider $S$ a finite set,  $P_0 = \{p_s\}_{s \in S}$, $(\sigma_{p_s})_{s \in S},(g_A(p_{sA}))_{s \in S}, (g_B(p_{sB}))_{s \in S}$, $(\lambda_{p_s})_{s \in S}$ and $\phi$ to be the limits of those sequences. It is straightforward to check that the limit of the sequences of solutions satisfy (3), (4) and (5).
The joint-plan is now defined as follows: consider as contracts $\gamma_{A}^{s} = m_{A}\sigma_{p_s} \bigotimes g_A(p_{sA})$, $\gamma_{B}^{s} = m_{B}\sigma_{p_s} \bigotimes g_B(p_{sB})$. 

Let $\tau^A_{p_s} = (g_A(p_{sA}))^{T}$ and $\tau^B_{p_s} = (g_B(p_{sB}))^{T}$. By construction, $\max\limits_{\sigma}\sigma A(p_{sA})\tau_{p_{s}}^A = v_A(p_{sA})$ and $\max\limits_{\sigma}\sigma B(p_{sB}) \tau^B_{p_s} = v_B(p_{sB})$. Therefore it implies that $(m_A \sigma_{p_s})A(p_{sA})\tau_{p_s}^A \leq v_A(p_{sA})$ and $(m_B\sigma_{p_s})B(p_{sB})\tau^B_{p_s} \leq v_B(p_{sB})$, for each $s \in S$. This implies in particular (1) of Lemma \ref{joint implies eq} is satisfied. Condition (3) of Lemma \ref{joint implies eq} follows directly from the fact that condition (5) of Theorem \ref{SST} above is satisfied by $\phi$. We can now define the signaling strategy $x$: let  $S \subset H^1_{t_0}$ with $|H^1_{t_0}| \geq |K_A \times K_B|$, for $t_0 \in \mathbb{N}$. Define $x^{k_A,k_B}(s) = \lambda_{p_s}\frac{p^{k_A,k_B}_s}{p^{0,k_A,k_B}}$. This signaling strategy satisfies condition (2) of Lemma \ref{joint implies eq} and induces the appropriate posteriors: in case $s \in S$ is observed, the uninformed players will update their priors according to Bayes rule to posterios $p_s$. 

We now calculate the ex-ante payoffs of player 1 obtained from the joint-plan just defined. Using condition (5) of Theorem \ref{SST}, we have that $\sum_{p_s}\lambda_{p_s} (h(\sigma_{p_s},p_s) \cdot p_s) = \phi \cdot p^0 = \text{Cav}(h)(p^0)$.  Rewriting the expression for $h$, $\sum_{p_s}\lambda_{p_s}(m_A\sigma_{p_s} A(p_{s_A}) \tau^A_{p_s} + m_B\sigma_{p_{s}}B(p_{s_B})\tau^B_{p_s})= \text{Cav}(h)(p^0)$. This proves the result for $p^0 \in \text{int} \Delta(K_A \times K_B)$. If now $p^0 \in \partial \Delta(K_A \times K_B)$, then consider the model $\mathcal{G}(\bar p)$ defined for $\bar p \in \text{int}\Delta(\supp(p^0))$, where $\bar p^{k_A, k_B} = p^{0,k_A,k_B}, \forall (k_A, k_B) \in \text{supp}(p^0)$ and apply the result proved to this case.\end{proof}

\subsection{Approachability Strategies}
\label{AS}

In this section we construct an approachability strategy for the uninformed player in game $G_{A+B}(p^0)$ and show that this strategy can indeed be played by players $2$ and $3$ in game $\mathcal{G}(p^0)$. The proof is simple but requires several preliminary definitions. This section is taken from Sorin \cite{SS2002} and adapted to our setting. Fix $p \in \Delta(K_A \times K_B)$. The main result is Proposition \ref{approach strategy}.

Let $C$ be a $|I_A \times I_B| \times |J_A \times J_B|$-matrix with coefficient in $\mathbb{R}^{\supp(p)}$, where $$ C_{i_A,j_A,j_A,j_B} = (C^{k_A, k_B}_{i_A,i_B,j_A,j_B})_{(k_A, k_B) \in \supp(p)}= (a^{k_{A}}_{i_{A},j_{A}} + b^{k_{B}}_{i_{B}, j_{B}})_{(k_A, k_B) \in \supp(p)},$$ where $a^{k_{A}}_{i_{A},j_{A}} \in \mathbb{R}$ and $b^{k_{B}}_{i_{B}, j_{B}} \in \mathbb{R}$. We define a \textit{vector payoff zero-sum game}: at stage $n$, player 1 (resp. player 2) chooses a move $(i^n_A,i^n_B)$(resp. $(j^n_A,j^n_B)$). The corresponding vector payoff $g_n = C_{i^n_A,j^n_A,j^n_A,j^n_B}$ is announced. Denote by $h_n$ the sequence of vector payoffs at stage $n$. This is the information available to both players up to stage $n$. Let $\overline{g}_n = \frac{1}{n}\sum_{m=1}^{n}g_m$ be the vector of average payoffs up to stage $n$. Let $||C|| = \max\limits_{i_A,i_B,j_A,j_B,k_A, k_B}|C^{k_A, k_B}_{i_A,i_B,j_A,j_B}|$.

\begin{definition}\label{approachable set}
A set $P \subset \mathbb{R}^{\supp(p)}$ is approachable by player 2 if for any $\epsilon>0$ there exists strategy $\tau$ of player 2 and $N \in \mathbb{N}$ such that for any strategy $\sigma$ of player 1 and $n \geq N$: $$\mathbb{E}_{\sigma,\tau}[d_n] \leq \epsilon,$$ where $d_n$ is the euclidean distance $d(\overline{g}_n, P)$.
\end{definition}

Let $$C\tau = \text{co}\{\sum_{j_A,j_B}C_{i_A, i_B, j_A, j_B}\tau_{j_A,j_B}  | (i_A, i_B) \in I_A \times I_B \},$$ where $C_{i_A, i_B, j_A, j_B}\tau_{j_A,j_B} := (C^{k_A,k_B}_{i_A, i_B, j_A, j_B}\tau_{j_A,j_B})_{(k_A,k_B) \in \supp(p)}$.
\begin{definition}

A closed set $P \subset \mathbb{R}^{\supp(p)}$ is a $B$-set for player 2 if for any $z \notin P$ there exists a closest point $y = y(z)$  in $P$ to $z$ and a mixed move $\tau = \tau(z) \in \Delta(J_A \times J_B)$, such that the hyperplane through $y$ orthogonal to the segment $[y,z]$ separates $z$ from $C\tau$.  
\end{definition}

\begin{theorem}\label{bset thm}
Let $P$ be a $B$-set for player 2. Then $P$ is approachable by that player. More precisely with a strategy satisfying $\tau(h_{n+1}) = \tau(\overline{g}_n)$, whenever $\overline{g}_n \notin P$, one has: $$ \mathbb{E}_{\sigma,\tau}[d_n] \leq \frac{2||C||}{\sqrt{n}}, \forall \sigma $$ and $d_n$ converges $ \mathbb{P}_{\sigma, \tau}$ a.s. to 0.
\end{theorem}

\begin{proof}
See Theorem B1 in \cite{SS2002}.\end{proof}

\begin{remark} The strategy $\tau$ obtained in the statement of Theorem \ref{bset thm} above will be called an \textit{approachability strategy}. \end{remark}

For $z \in \mathbb{R}^{\supp(p)}$, let $M(z):= z - \mathbb{R}^{\supp(p)}_{+}$.

\begin{proposition}\label{approach strategy}
Let  $ z \in Z_{\infty} = \{ z \in \mathbb{R}^{\supp(p)} | z \cdot q \geq \mathfrak{h}(q), \forall q \in \Delta(\supp(p)) \}$. Then player 2 can approach $M(z)$. Also, the approachability strategy $\tau$ for player 2 can be assumed to satisfy $\tau(h_t) = \tau^A(h_t) \bigotimes \tau^B(h_t)  \in \Delta(J_A) \bigotimes \Delta(J_B), \forall h_t \in H_t,  t \in \mathbb{N}$.
\end{proposition}

\begin{proof}
Approachability of $M(z)$ follows from Theorem 3.33 in \cite{SS2002}, where it is checked that $M(z)$ is a $B$-set. Fix $\overline{g}_n$ and let $\tau(\overline{g}_n)$ be the mixed move associated with the $B$-set $M(z)$ and $\overline{g}_n$ in Theorem \ref{bset thm}. Let $\tau^A = \text{marg}_{J_A}\tau$ and $\tau^B = \text{marg}_{J_B}\tau$. We show $C \tau = C(\tau^A \bigotimes \tau^B)$: let $x \in C \tau$ such that $x := \sum_{(i_A, i_B)}\lambda^{i_A,i_B}\sum_{j_A,j_B}C_{i_A, i_B, j_A, j_B}\tau_{j_A,j_B}$ with $\sum_{i_A, i_B}\lambda^{i_A,i_B} = 1$ and $\lambda^{i_A,i_B} \geq 0$. Then,
\begin{center}
$ \sum_{(i_A, i_B)}\lambda^{i_A,i_B}\sum_{j_A,j_B}C_{i_A, i_B, j_A, j_B}\tau_{j_A,j_B} = \sum_{(i_A, i_B)}\lambda^{i_A,i_B} \sum_{j_A,j_B} (a^{k_A}_{i_A,j_A} + b^{k_B}_{i_B, j_B})_{(k_A,k_B)} \tau_{j_A,j_B} = \sum_{i_A, i_B}\lambda^{i_A,i_B} [\sum_{j_A} a^{k_A}_{i_A,j_A}\tau^A_{j_A} + \sum_{j_B}b^{k_B}_{i_B, j_B}\tau^B_{j_B}]_{(k_A,k_B)} = \sum_{i_A, i_B}\lambda^{i_A,i_B} \sum_{j_A,j_B}(a^{k_A}_{i_A,j_A} + b^{k_B}_{i_B, j_B})_{(k_A,k_B)}\tau^A_{j_A} \tau^B_{j_B} = \sum_{i_A, i_B}\lambda^{i_A,i_B}\sum_{j_A,j_B}C_{i_A, i_B, j_A, j_B}\tau^A_{j_A} \tau^B_{j_B},$
\end{center}
which implies that $x \in C(\tau^A \bigotimes \tau^B)$. Now, let $y := \sum_{(i_A, i_B)}\lambda^{i_A,i_B}\sum_{j_A,j_B}C_{i_A, i_B, j_A, j_B}\tau^A_{j_A} \tau^B_{j_B} \in C(\tau^A \bigotimes \tau^B)$, with $\sum_{(i_A, i_B)}\lambda^{i_A,i_B} = 1$ and $\lambda^{i_A,i_B} \geq 0$. Using the the same equalities above, it implies $y \in C\tau$.  
So, if player 2 uses $\tau_t(h_t) : = \tau^A(h_t) \bigotimes \tau^B(h_t)$ after each $h_t$ , $\tau$ is an approachability strategy. \end{proof}

\subsection{A Continuum of Equilibria in $\mathcal{G}(p^0)$}

In subsection 3.1 of the main paper we showed how $\Cav(v_A)(p^0_A) + \Cav(v_B)(p^0_B)$ is an equilibrium of $\mathcal{G}(p^0)$, provided some conditions are satisfied. Proposition \ref{existence lower bound} in this Appendix shows how $\Cav(h)(p^0)$ is always an equilibrium payoff of the informed player in $\mathcal{G}(p^0)$. The next proposition immediately implies that any element in $I(p^0)$ is an ex-ante equilibrium payoff of the informed player, implying a continuum of equilibria exists in the model $\mathcal{G}(p^0)$.

\begin{proposition}\label{continuum}
Let $(\s,\t_A,\t_B)$ and $(\s', \t'_A, \t'_B)$ be two uniform equilibria of $\mathcal{G}(p^0)$ with associated ex-ante payoffs $\g$ and $\g'$ for the informed player, respectively. For any $\a \in (0,1)$, there exists a uniform equilibrium $(\bar \s, \bar \t_A, \bar \t_B)$ whose ex-ante equilibrium payoff to the informed player is $\a \g + (1-\a) \g'$.
\end{proposition}

\begin{proof}Fix $\a \in (0,1)$. We show that $\a \g + (1-\a)\g'$ is an ex-ante equilibrium payoff for the informed player in $\mathcal{G}(p^0)$.  Since each of the uninformed players can play their optimal strategy in their repeated zero-sum games $G_A(p^0_A)$ or $G_B(p^0_B)$, it implies that $$\text{Cav}(v_A)(p^0_A)+\text{Cav}(v_B)(p^0_B) \geq \g$$ as well as $$\text{Cav}(v_A)(p^0_A)+\text{Cav}(v_B)(p^0_B) \geq \g'.$$ Now, since the informed player has an optimal strategy in the zero-sum game $G_{A+B}(p^0)$ that guarantees him $\text{Cav}(h)(p^0)$, it implies that $\g \geq \text{Cav}(h)(p^0)$ as well as $\g' \geq \text{Cav}(h)(p^0)$. Consider a jointly controlled lottery that implements the equilibrium profile associated with $\g$ with probability $\a$ and the equilibrium associated with $\g$ with probability $1-\a$. By the properties of the jointly controlled lottery, there cannot be profitable undetectable deviations at the stages where the jointly controlled lottery is played. For detectable deviations of the uninformed player 2 (respectively, player 3) at the lottery stages, the informed player plays the optimal strategy of the zero-sum game $G_A(p^0_A)$ (respectively $G_B(p^0_B)$) to punish. For detectable deviations of the informed player at the lottery stages, the uniformed players play the approachability strategy of Theorem \ref{approach strategy} to punish. The strategy profile where a jointly controlled lottery is played at initial stages -- with deviations punished as described --  and, after that, the corresponding strategy profile paying $\g$ or $\g'$ drawn from the lottery, is a uniform equilibrium of the game $\mathcal{G}(p^0)$. Indeed, we already showed that there are no profitable deviations during lottery stages. After lottery stages, players play a uniform equilibrium so there is no profitable deviation for any player, also. The ex-ante payoff of this equilibrium is $\a \g +  (1-\a)\g'$.\end{proof}


\begin{thebibliography}{99}
\bibitem{AMS1995}
Aumann, R.J., Maschler, M., Stearns, R.E. (1995): {\it Repeated Games with Incomplete Information}. MIT Press. 

\bibitem{AB2019}
Arieli, I., Babichenko, Y. (2019): ``Private Bayesian Persuasion,'' {\it Journal of Economic Theory}, 182: 185-217.

\bibitem{DB1956}
Blackwell, D. (1956): ``An analog of the minmax theorem for vector payoffs,'' {\it Pacific Journal of Mathematics} 6(1): 1-8.

\bibitem{FC1975}
Clark, F. H. (1975): ``Generalized Gradients and Applications,'' {\it Transaction of the American Mathematical Society}, vol. 205, pp. 247-272.

\bibitem{LSR2010}
De Loera, J., Rambau, J., Santos, F. (2010) {\it Triangulations: Structures for Algorithms and Applications}. Springer.

\bibitem{FF1992}
Forges, F. (1992): ``Repeated Games of Incomplete Information: non-zero-sum,'' {\it Handbook of Game Theory with Economic Applications},1 : 155-177. 

\bibitem{FS2015}
Forges, F., Solomon, A. (2015): ``Bayesian Repeated Games and Reputation,'' {\it Journal of Economic Theory}, 70-104.

\bibitem{FHS2016}
Forges, F., Horst, U., Solomon, A. (2016): ``Feasibility and individual Rationality in two-person Bayesian Games,'' {\it International Journal of Game Theory}, 45:11-36.

\bibitem{KG2011}
Gentzkow, M., Kamenica, E. (2011): ``Bayesian Persuasion,'' {\it American Economic Review }, 101 (6):2590-2615. 

\bibitem{SH1985}
Hart, S. (1985): ``Nonzero-sum Two-person Repeated Games with Incomplete Information,'' {\it Mathematics of Operations Research}, vol. 10, number 1, 117-153.

\bibitem{HL2021}
Huangfu, B., Liu, H. (2021): ``Information Spillover in Multi-good adverse selection,'' {\it American Economics Journal: Microeconmics}, forthcoming. 

\bibitem{KLT2022}
Koessler, F., Laclau, M., Tomala, T. (2022): ``Interactive Information Design,'' {\it Mathematics of Operations Research}, 47(1), 153-175.

\bibitem{LRS2019}
Laraki, R., Renault, J., Sorin, S. (2019): {\it Mathematical Foundations of Game Theory}. Springer-Verlag.

\bibitem{RL2004}
Laraki, R. (2004): ``Regularity of the Convexification Operator on a Compact Set'', {\it Journal of Convex Analysis} 11(1): 209-234.

\bibitem{MSZ2013}
Maschler, M., Solan, E., Zamir, S. (2013): {\it Game Theory}. Cambridge University Press.

\bibitem{JS1988}
Shalev, J. (1988): {\it Nonzero-sum two-person repeated games with incomplete information and observable payoffs.} Tel Aviv University. Faculty of Management. 

\bibitem{SST1995}
Simon, R.S., Spiez, S., Tor\'unczyk, H. (1995) ``The existence of equilibria in certain games, separation for families of convex functions and a theorem of Borsuk-Ulam type,'' {\it Israel Journal of Mathematics}, 92, 1-21.

\bibitem{SS1983}
Sorin, S. (1983):  ``Some results on the Existence of Nash Equilibria of Non-Zero Sum Games with Incomplete Information,'' {\it International Journal of Game Theory}, Vol. 12, Issue 4, p. 193-205. 

\bibitem{SS2002}
Sorin, S. (2002): {\it A First Course on Zero-Sum Repeated Games}. Math\'ematiques et Applications 37. Springer. 

\bibitem{TR1970}
Rockafellar, T. (1970): {\it Convex Analysis}. Princeton University Press. 

\bibitem{YW2013}
Wang, Y. (2013): ``Bayesian Persuasion with Multiple Receivers,'' {\it Mimeo}.

\end{thebibliography}
\end{document}